%% file: main.tex
\documentclass[11pt]{article} 

\usepackage[
  margin=1in,        
  includefoot=false, 
  footskip=30pt      
]{geometry}
\usepackage{fancyhdr}
\usepackage{amsmath,amsfonts, amsthm, amssymb}
\usepackage{tikz}
\usepackage{enumerate}
\usepackage{graphicx}
\usepackage[normalem]{ulem}
\usepackage{color}
\usepackage{braket}
\usepackage{mathtools}
\usetikzlibrary{decorations.pathreplacing}
\usepackage{hyperref}
\usepackage{thmtools, thm-restate}
\usepackage{authblk}
\newtheorem{theorem}{Theorem} 
\newtheorem{mainthm}[theorem]{Theorem}

\newtheorem{lemma}[mainthm]{Lemma}
\newtheorem{conjecture}[mainthm]{Conjecture}
\DeclareFontFamily{U}{mathx}{\hyphenchar\font45}
\DeclareFontShape{U}{mathx}{m}{n}{
      <5> <6> <7> <8> <9> <10>
      <10.95> <12> <14.4> <17.28> <20.74> <24.88>
      mathx10
      }{}
\DeclareSymbolFont{mathx}{U}{mathx}{m}{n}
\DeclareFontSubstitution{U}{mathx}{m}{n}
\DeclareMathAccent{\widecheck}{0}{mathx}{"71}
\usepackage{scalerel,stackengine}

\def\ket#1{\mathinner{|{#1}\rangle}}

\def\Real{{\mathbf R}}

\def\Complex{{\mathbb C}}

\def\Prob{{\mathbf P}}

\def\Expec{{\mathbf E\,}}

\def\eps{{\varepsilon}}
\def\One{{\mathbf{1}}}

\newcommand{\poly}{{\mathrm{poly}}}

\newcommand{\diff}{\mathop{}\!\mathrm{d}}
\numberwithin{equation}{section}

\newcommand{\mcal}[1]{{\mathcal{#1}}}

\newcommand{\BPP}{{\mathsf{BPP}}}
\newcommand{\NP}{{\mathsf{NP}}}
\newcommand{\sharpP}{\#\mathsf{P}}

\DeclareMathOperator{\trace}{tr}

\DeclareMathOperator{\Per}{Per}

\newtheorem{definition}[theorem]{Definition}

\theoremstyle{remark}
\newtheorem{remark}[theorem]{Remark}

\definecolor{MidnightBlue}{RGB}{25,25,150}
\definecolor{BrickRed}{RGB}{182,50,28}
\definecolor{ForestGreen}{RGB}{34,139,34}

\begin{document}

\title{Exponential improvements to the average-case hardness of BosonSampling}

\date{}

\clearpage

\author[1]{Adam Bouland\thanks{ \texttt{abouland@stanford.edu}}}
\author[1]{Ishaun Datta\thanks{ \texttt{idatta@stanford.edu}}}
\author[2]{Bill Fefferman\thanks{ \texttt{wjf@uchicago.edu}}}
\author[3]{Felipe Hern\'{a}ndez\thanks{ \texttt{felipeh@psu.edu}}}

\affil[1]{Department of Computer Science, Stanford University}
\affil[2]{Department of Computer Science, University of Chicago}
\affil[3]{Department of Mathematics, Penn State}

\maketitle
\thispagestyle{empty}

\begin{abstract}

BosonSampling and Random Circuit Sampling are important both as a theoretical tool for separating quantum and classical computation, and as an experimental means of demonstrating quantum speedups. Prior works have shown that average-case hardness of sampling follows from certain unproven conjectures about the hardness of computing output probabilities, such as the Permanent-of-Gaussians Conjecture (PGC), which states that $e^{-n\log{n}-n-O(\log n)}$ additive-error estimates to the output probability of most random BosonSampling experiments are \textsf{\#P}-hard. Prior works have only shown weaker average-case hardness results that do not imply sampling hardness. Proving these conjectures has become a central question in quantum complexity.

In this work, we show that $e^{-n\log n -n - O(n^\delta)}$ additive-error estimates to output probabilities of most random BosonSampling experiments are \textsf{\#P}-hard for any $\delta>0$, exponentially improving on prior work. In the process, we circumvent all known barrier results for proving PGC. The remaining hurdle to prove PGC is now ``merely'' to show that the $O(n^\delta)$ in the exponent can be improved to $O(\log n).$ We also obtain an analogous result for Random Circuit Sampling.

We then show, for the first time, a hardness of average-case classical \textit{sampling} result for  BosonSampling, under an anticoncentration conjecture. Specifically, we prove the impossibility of multiplicative-error sampling from random BosonSampling experiments with probability $1-2^{-\tilde O(N^{1/3})}$ for input size $N$, unless the Polynomial Hierarchy collapses. This exponentially improves upon the state-of-the-art. To do this, we introduce new proof techniques which tolerate exponential loss in the worst-to-average-case reduction. This opens the possibility to show the hardness of average-case sampling without ever proving PGC.

\end{abstract}

\newpage\thispagestyle{empty}\tableofcontents

\newpage \section{Introduction}
\setcounter{page}{1}

What makes quantum mechanics hard to simulate classically?  This has been the central question of quantum computation since it was first proposed \cite{benioff1980computer, manin1980computable,feynman1982simulating}. 
The need to answer this question has been made even more urgent by recent experiments claiming to solve certain problems much faster than by any classical computer
\cite{supremacy,google2,ustcrcs,ustc-bosonsampling1,ustc-bosonsampling2,ustc-bosonsampling3,nist-bosonsampling,xanadu,quantinuum,liu2025robust}. 
These claims represent the first \emph{experimental} violations of the Extended Church-Turing Thesis, the belief that all physically realizable models of computation are efficiently simulable by randomized Turing machines, and thus deserve careful scrutiny. 

While the physics of these experiments differs dramatically, from a computational standpoint they all solve random sampling problems 
that have three components: \textit{(i)} initialize a fiducial starting state (such as $\ket{0^n}$), \textit{(ii)} evolve by a random quantum circuit drawn from some particular distribution, and \textit{(iii)} measure to draw a sample. 
Seminal results due to Terhal and DiVincenzo \cite{terhal2004adptive} and concurrent works of Aaronson and Arkhipov \cite{Aaronson2013} and Bremner, Jozsa, and Shepherd \cite{bremner2010iqp} gave evidence that for these distributional sampling problems, even sub-universal quantum computation can outperform any efficient classical algorithm in the worst case.
Fascinatingly this only assumes the non-collapse of \textsf{PH}.

However, the shortcoming of these statements is that they are \emph{brittle}, pertaining only to exact sampling in the worst case. Thus the important open problem is to make these separations more robust, so as to close the gaps between theory and experiment.

Foremost among these gaps is to prove classical hardness of sampling from an \emph{average-case} experiment, i.e.,\ to extend the previous worst-case sampling results to prove that sampling is hard for a \emph{randomly chosen} or typical circuit from a given ensemble. This random choice of circuit is crucial in quantum advantage experiments.  For one thing, such randomness gives a hard candidate distribution to test with a quantum device.  
Additionally, randomness plays an important role in classical verification, e.g.,\ by benchmarking tests such as Linear Cross-Entropy which make use of specific properties of random circuits.
It has also been shown that average-case hardness allows one to derive cryptographic primitives (see e.g., \cite{khurana2024foundingquantumcryptographyquantum}).
Thus for reasons both practical and fundamental, it behooves us to study the average-case hardness of sampling---which, in contrast to the worst case, remains an \emph{open problem} for \emph{all} quantum advantage proposals. 

Why should sampling from random quantum circuits be intractable for classical computers?  The first evidence for this came from Aaronson and Arkhipov, in two parts.
First, they showed a reduction from the problem of \emph{sampling} from random circuits to the problem of \emph{approximately computing} output probabilities of random circuits, via Stockmeyer counting \cite{stockmeyer1983complexity}.
Second, they showed that computing an output probability of a random BosonSampling circuit is just as hard as in the worst case, i.e.\ $\#\mathsf{P}$-hard, using the connection between bosons and the matrix permanent.  
This was extended to random circuits  on qubits \cite{bouland_complexity_2019} and subsequently improved and generalized \cite{movassagh2023hardness,haferkamp_closing_2019}.
However, these two parts do not connect with one another to establish hardness of sampling.
The key issue is that existing average-case hardness of computing proofs are not \emph{error tolerant} enough to prove hardness of sampling.
That is, to show hardness of sampling we want to show it is $\#\mathsf{P}$-hard to estimate output probabilities to additive error $\epsilon$, but so far we have only proven it is $\#\mathsf{P}$-hard to estimate them to additive error $\epsilon' \ll \epsilon$.
We call this gap between $\epsilon$ and $\epsilon'$ the ``robustness gap'', and it remains open for all quantum advantage proposals.
Hardness of average-case sampling has therefore only been established under unproven conjectures.

Consequently, 
the focus of this work is to close this robustness gap.
While our results are broadly applicable to many random sampling experiments, we will primarily focus on BosonSampling. The goal is to prove the \textit{Permanent-of-Gaussians Conjecture} (PGC), the statement that the following problem, known as Gaussian Permanent Estimation, \textsf{GPE}$_{\pm},$ is $\sharpP$-hard: estimate the output probability of a random BosonSampling experiment to within additive error $\pm e^{-n\log{n}-n-O(\log n)}.$ 
Aaronson and Arkhipov showed this conjecture suffices to show hardness of sampling from average-case BosonSampling experiments \cite{Aaronson2013}.\footnote{We note \cite{Aaronson2013} also conjecture a certain ``flatness'' property about the output distribution known as anticoncentration (which we also assume in this work). This allows them to convert average-case additive estimates to relative-error estimates.}

In the last decade, progress has been made toward proving PGC \cite{Aaronson2013,Bouland2021,krovicomm}. 
While Aaronson and Arkhipov's initial work showed computing additive error estimates of $e^{-O(n^4)}$ to the output probability of most BosonSampling experiments is \textsf{\#P}-hard \cite{Aaronson2013}, this error tolerance was subsequently improved to $e^{-6n\log{n}-O(n)}$ by Bouland, Fefferman, Landau, and Liu  \cite{Bouland2021}, and then to $e^{-4n\log{n} -O(n)}$ in unpublished work of Krovi \cite{krovicomm}.  Therefore the remaining gap to establish the hardness of BosonSampling is to improve the robustness of this result by a constant factor in the exponent.  This seems tantalizingly ``close’’ to the target in additive terms yet \textit{exponentially far away} in relative terms.
We note the analogous conjectures for all other quantum advantage experiments remain open as well, such as Random Circuit Sampling \cite{Boixo2016}, despite much progress in the area \cite{bouland_complexity_2019,movassagh2023hardness,Bouland2021,Kondo2021_robustness,fermionsampling,krovi2022}.

Why has it been so difficult to improve the robustness of output probability estimation and prove the classical hardness of BosonSampling or any other quantum advantage experiment?  One of the major reasons is that there are a number of proof barriers that have been identified, indicating that to prove hardness of sampling, new techniques are required:\footnotemark 
\begin{itemize}

    \item \textit{Convexity barrier.} Noted in \cite{Aaronson2013}, the basic idea is that worst-to-average-case reductions for the permanent are based on polynomial extrapolation, following Lipton's proof \cite{lipton1991}. 
    Polynomial extrapolation is in general exponentially ill-conditioned, i.e.\ an error $\gamma$ in a degree-$d$ polynomial $p(t)$ near $t=0$ becomes error $\sim2^d \gamma$ near $t=1$.
    Moreover, one can show this is necessary even for the set of polynomials corresponding to valid matrix permanents, which is a convex set. 
    Thus any worst-to-average case reduction for the permanent based on polynomial extrapolation will introduce exponential relative error. A special case was referred to as the ``noise barrier'' of \cite{Bouland2021}.

\item \textit{``Jerrum-Sinclair-Vigoda'' barrier for BosonSampling.}
This barrier is inspired by a landmark result of \cite{jerrum-sinclair-vigoda} giving an efficient classical algorithm to estimate the permanent of a nonnegative matrix to $1/\poly(n)$ relative error.  This algorithm tells us that any technique used to prove PGC must fundamentally make use of the fact that i.i.d.\ Gaussian matrices have negative as well as positive entries.
By contrast, all existing worst-to-average-case reductions for Gaussian permanents work equally well for permanents of nonnegative matrices, and therefore cannot possibly prove PGC.
In other words, to show hardness of sampling, we will need a proof which uses a special property of matrices with negative entries that does not hold for nonnegative matrices, such as multiplicative hardness in the worst case.

\item \textit{Depth and ``Born-rule'' barriers for Random Circuit Sampling.} \cite{Napp2020} gives a classical algorithm that approximately samples from the output distribution of a particular ensemble of constant depth RCS experiments.  On the other hand, the existing techniques for proving hardness of computing output probabilities work with respect to circuits of any depth.  Therefore, if we are to prove hardness of sampling, we need to find a proof technique that is sensitive to circuit depth and only works to prove hardness for sufficiently deep circuits. 

The ``Born-rule'' barrier identified by Krovi \cite{krovi2022} is that the additive error needed to prove the hardness of average-case sampling ($~2^{-n}$) is already larger than the additive error known to be hard in the worst case ($2^{-2n}$, which is derived from the Born rule by squaring the output amplitude of a Quantum Fourier Sampling circuit).
How can we ever hope to prove a worst-to-average case reduction in which the additive error in the average case \emph{is larger} than the additive error we need to obtain in the worst case?  
\footnotetext{
There is also a relativization barrier to proving hardness of average-case sampling to small $\ell_1$ error \cite{Aaronson2017foundations}. However, here we focus on average-case multiplicative-error sampling (Def.\ \ref{def:sampler}) to which no relativization barrier applies.
}
\end{itemize}

There has also been work aiming to falsify variants of PGC. 
For example, Eldar and Mehraban showed there is a quasipolynomial-time classical algorithm to multiplicatively estimate random Gaussian permanents if the means are non-zero but asymptotically slowly vanishing, despite being $\sharpP$-hard to compute exactly on average \cite{eldar-mehraban}.
Thus the dividing line between classically easy and classically hard is very narrow---making it yet more difficult to furnish a proof of PGC.

\subsection{Our results}

In this work we introduce a new suite of tools which allow us to exponentially improve on the state-of-the-art hardness results for BosonSampling. 
In particular, we invent new techniques that overcome \emph{all of the barriers} described above.

Our first result makes progress towards proving the Permanent-of-Gaussians Conjecture (PGC).
We show a new worst-to-average-case reduction for computing Gaussian permanents whose additive error tolerance exponentially improves on the state-of-the-art.
Our error tolerance for the first time matches to leading order that of the Permanent-of-Gaussians Conjecture (PGC).

\begin{restatable}[Hardness of computing output probabilities]{theorem}{dilutionrestatable}
\label{thm:dilution-intro}
            For any $\delta>0,$ it is \textsf{\#P}-hard under a $\mathsf{BPP}^{\mathsf {NP}}$ reduction to approximate output probabilities of an $n$-photon, $O(n^2)$-mode BosonSampling experiment to additive error $\exp(-n\log n -n - O(n^\delta))$    with success probability at least $ 2/3$, assuming the Permanent Anticoncentration Conjecture \ref{conj:pacc-aa}.
\end{restatable}

\noindent This is nearly the additive error tolerance needed to prove PGC, $e^{-n\log{n}-n-O(\log n)}$. In particular, all that remains is ``merely'' to improve the $O(n^\delta)$ term in the exponent to $O(\log n)$.
In order to prove this result, we give a new worst-to-average-case reduction for BosonSampling which replaces polynomial extrapolation with polynomial \emph{coefficient extraction}.
This allows us to use a technique we call \emph{``dilution''} to lessen the degree of the polynomial involved 
and hence reduce the error blowup of the worst-to-average-case reduction.

Crucially, our proof surpasses the Jerrum-Sinclair-Vigoda barrier as it requires that the worst-case matrix contain both positive and negative entries.
This is because our worst-to-average-case reduction derives the worst-case permanent value to within small relative error, which is only $\#\mathsf{P}$-hard with mixed signs.
This is an essential feature of any proof that might solve PGC, and a feature which was missing from all prior proofs of average-case hardness for the permanent \cite{Aaronson2013,Bouland2021,krovicomm}.

We also show this idea can be ported to other quantum advantage experiments, like RCS:
\begin{restatable}{corollary}{rcsdilution}
\label{cor:rcs-dilution}
        For any $\delta>0$, it is $\#\textsf{P}$-hard to approximate the output probabilities of $n$-qubit Random Circuit Sampling experiments of $\Omega(\log n)$ depth to additive error $2^{-n- O(n^\delta)}$.
\end{restatable}

Just as with BosonSampling, this exponentially improves over prior work \cite{bouland_complexity_2019,movassagh2023hardness,Bouland2021,Kondo2021_robustness,krovi2022}, and obtains hardness which is within an $O(n^\delta)$ factor of what is needed for hardness of sampling.
Applied to RCS, our techniques overcome the depth barrier by requiring anticoncentration, and the Born-rule barrier by ``diluting'' the worst-case instance to be polynomially smaller than the average-case instance. We prove Corollary \ref{cor:rcs-dilution} in Sec.\ \ref{sec:rcs}.

Our second major contribution is to develop a suite of techniques that taken together allow us to show, for the first time, a hardness of average-case sampling theorem.

\begin{restatable}[Hardness of sampling]{theorem}{nosampler}
\label{thm:no-sampler}
    There does not exist a multiplicative-error classical sampler
    (see Def.\ \ref{def:sampler}) 
    from the output distribution of an $n$-photon, $O(n^2)$-mode real BosonSampling experiment that succeeds with probability at least~$1-\exp(-O(n))$ over the choice of experiment, assuming $\textsf{PH}$ does not collapse and a slight generalization of Permanent Anticoncentration, Conjecture \ref{conj:anticoncentration}.
\end{restatable}

This theorem exponentially improves upon the trivial hardness of sampling statement. In particular, if the sampling algorithm succeeds with probability $1-2^{-\tilde{O}(n^3)}$, then the algorithm directly estimates the value of the worst case as the input size\footnote{Here $n$ is the number of photons, and the input is an $n\times n$ matrix of reals specified to $\tilde{O}(n)$ bits of precision.} is $\tilde{O}(n^3)$ (see Lemma \ref{lem:nontriviality}). We note another exponential improvement would be required to show the desired hardness of sampling for $1-1/\poly(n)$ fraction of experiments. 
This is the first hardness result for average-case multiplicative-error sampling.
This had been open for all quantum advantage proposals, as prior hardness results for \emph{computing} output probabilities do not imply  average-case \emph{sampling} hardness (even for exact sampling) due to the losses in the Stockmeyer reduction from sampling to computing. 

In order to show this result, we develop a suite of new techniques that allow us to tolerate an exponential error blowup in the worst-to-average-case reduction, overcoming the convexity barrier of \cite{Aaronson2013}.
This is achieved by ``magnifying'' the worst-case permanent value to tolerate more error in the reduction, among other improvements.

This still falls short of proving PGC---the bottleneck is that the average-case algorithm can only compute permanents of matrices which are close in total variation distance to i.i.d. Gaussian, which limits the error tolerance.
To overcome this bottleneck, we show that if an average-case algorithm works with sufficiently high probability, then it can also compute permanents ``out of distribution'' in TV distance. 
This uses special properties of the Gaussian measure and also requires proving new results in random matrix theory regarding submatrices of Haar-random orthogonals.
Our work opens up the possibility that one could prove the classical hardness of sampling, even \emph{without} proving PGC, by improving some of the parameters of these new tools.
Interestingly, this result pertains only to real BosonSampling, and extending to the complex case requires solving an open problem in complex analysis---see Appendix \ref{sec:complex-squares}.

\subsection{Proof techniques}
\subsubsection{What controls robustness in the standard worst-to-average-case reduction?}\label{subsubsec:what-controls-robustness}
To explain our proof, it is helpful to briefly recall the average-case hardness proofs of \cite{Aaronson2013} and its subsequent improvements \cite{Bouland2021,krovicomm}. The basic idea is to use polynomial extrapolation to show the squared permanent is hard to compute on average, following Lipton \cite{lipton1991}.
Suppose we wish to compute the squared permanent of an arbitrary (worst-case) matrix $W\in\{0,\pm 1\}^{n\times n}$ using only the ability to compute most Gaussian permanents $R$ drawn from $\mcal N(0,1)^{n\times n}$.
We define a univariate family of matrices interpolating between $W$ and a single random choice of Gaussian matrix $R$:
\[A(t) = (1-t)R + tW\]
This family has three nice properties that enable the reduction: \textit{(i)}  $|\Per(A(t))|^2$ is a degree $2n$ polynomial in $t$, \textit{(ii)}  for small values of $t,$ $A(t)$ is close to i.i.d.\ Gaussian in total variation distance, and \textit{(iii)} $|\Per(A(1))|^2 = |\Per(W)|^2$.
This motivates a worst-to-average-case reduction whereby one computes $|\Per(W)|^2$ by computing the average-case permanents $|\Per(A(t))|^2$ at many small values of $t$, inferring the polynomial in $t$, and extrapolating it to $t=1$. This shows that computing average-case permanents, namely estimating the polynomial close to $t=0$, is as hard as computing a worst-case permanent, the polynomial at $t=1$.

What controls robustness, i.e.\ the additive error tolerance, in this worst-to-average-case reduction? In other words, what are the largest error bars we can tolerate on our estimates to the polynomial close to $t=0,$ and how do these errors accrue under polynomial extrapolation? 

Polynomial extrapolation is ill-conditioned, in the sense that errors in the values of the polynomial close to $t=0$ blow up exponentially under extrapolation to $t=1.$ Formally, we can quantify the error blowup 
using a discrete version of the Remez inequality that we prove in this work.\footnote{The Remez inequality is more commonly shown in a continuous form to bound sup norms of polynomials defined over measurable sets \cite{remez1936propriete}. We discretize the inequality so that it is more natural for a computer science setting and in particular our worst-to-average-case reduction.}  

\begin{restatable}[Discrete Remez inequality]{lemma}{discreteremez}
\label{lem:discrete-remez}
Let $\{x_j\}_{j=0}^{d}\subset [0,1]$ be a $\delta$-separated set of points, meaning that
$|x_i-x_j| \geq \delta$ for $i\not=j$.  Then if $p$ is a degree-$d$ polynomial and $L\geq1$,
\[
|p(L)| \leq (e^2(\delta d)^{-1} L)^{d} \max_{0\leq j\leq d} |p(x_j)|.
\]
\end{restatable}

\begin{table}[b!]
    \centering
    \begin{tabular}{|p{7.5cm}|c|c|}
        \hline
        \textbf{Technique} & \textbf{Polynomial degree} & \textbf{Box size $\Delta$} \\
        \hline
        \raggedright Robust Berlekamp-Welch \cite{Bouland2021} & $2n$ & $1/n^2$ \\
        \raggedright Tighter TVD analysis \cite{krovicomm} & $2n$ & $1/n$ \\
        \raggedright \textit{(This work)} Square method, Lemma \ref{lem:square-trick-coex}& $n$ & $1/n$ \\
         \raggedright \textit{(This work)} Rare events lemmas \ref{lem:probability-attenutation} \& \ref{prp:small-small}
        & $n$ & $1/\sqrt{n}$ \\
        \raggedright \textit{(This work)} Dilution via coefficient extraction, Thm.\ \ref{thm:dilution-intro} & $n^\delta$ $\forall$ const.\ $\delta>0$ & $1/n^\delta$ $\forall$ const.\ $\delta>0$ \\
        \hline
    \end{tabular}
    \caption{Lemma \ref{lem:discrete-remez} tells us that estimating a degree $d$ polynomial to within $\pm \gamma$ at points $t\in[0,\Delta]$ incurs a blowup at $t=1$ of $\gamma (1/\Delta)^d.$ Our work introduces a suite of techniques, shown here, that decrease $d$ and increase $\Delta.$ }
    \label{tab:techniques}
\end{table}

Lemma \ref{lem:discrete-remez} has a very simple interpretation. If we take $p(t)$ to be the difference between the true polynomial $|\Per(A(t))|^2$ and the approximate polynomial obtained from estimates of the permanent at $\{t_j\},$ then $\max_{0\leq j\leq d} |p(t_j)|\eqqcolon \gamma$ is precisely the robustness, i.e.\ the maximum additive error tolerance on average-case values of the permanent. Then Lemma \ref{lem:discrete-remez} tells us that the error blowup $|p(1)|$ is bounded above by $\gamma(\delta d)^{-d},$ which for $\delta$-separated points in the interval $[0,\Delta]$ where $0<\Delta<1$ is $\gamma (1/\Delta)^d.$ We will refer to $\Delta$ as the \textit{``box size,''} which is determined by the largest value of $t$ such that the total variation distance between $A(t)$ and i.i.d.\ Gaussian is, say, $0.01.$ In short, estimating a degree $d$ polynomial to within $\pm \gamma$ at points $t\in[0,\Delta]$ incurs a blowup at $t=1$ of $\gamma (1/\Delta)^d.$ 

Posed in this way, we see that to increase the robustness of our worst-to-average-case reduction we need to reduce our effective polynomial degree $d$ or increase the box size $\Delta$ over which we estimate average-case values. 
In both the proofs of \cite{Bouland2021} and \cite{krovicomm} the main improvement was in reducing the distance of extrapolation, while keeping the same degree of polynomial ($2n$ for a squared permanent).
In particular in \cite{Bouland2021} the distance was reduced to $\Delta = O(1/n^2)$ by introducing a robust version of Berkelamp-Welch over the complex numbers.\footnote{We note similar results for BosonSampling could be obtained by the techniques of \cite{Kondo2021_robustness}.}
In \cite{krovicomm} the box size was improved to $\Delta=O(1/n)$ by a more sophisticated calculation of the total variation distance between $A(t)$ and Gaussian, which saves a factor of $n^{2n} = e^{2n\log n }$.

In this work, we will introduce a suite of techniques that improve both the polynomial degree $d$ and the ``box size'' $\Delta.$ We tabulate these techniques and their improvements in Table \ref{tab:techniques}. 

\subsubsection{Coefficient extraction: a new way to encode the permanent}

A natural approach to try to improve the robustness of this argument is to reduce the degree of the polynomial involved.
A simple observation is that for any $\varepsilon>0$, it is $\#\mathsf{P}$-hard to compute the permanent of an $n^\eps \times n^\eps$ matrix $W$ as well---this is simply polynomially shrinking the input size.
Therefore a natural way to improve the robustness is to try to make $W$ smaller, an idea we henceforth refer to as \emph{``dilution.''}
Using standard polynomial extrapolation arguments, this doesn't yield much progress.
That's because if we set $W$ to have small support---say with only $O(n^\varepsilon)$ nonzero entries---then $\Per(W)=0$. Trivially, a matrix must have at least $n$ non-zero entries for its permanent to be non-zero. This lower bounds how much one could gain by such arguments using extrapolation, and the best one can obtain by dilution is $e^{-3n\log n -O(n)}$ robustness\footnote{This is obtained by setting $W$ to be a (tiny) arbitrary matrix of size $n^\varepsilon\times n^\varepsilon$ in direct sum with an identity on the remaining $n-n^\varepsilon$ dimensions.}---which sits right at the convexity barrier.

To get around this obstacle, our first step is to change the worst-to-average-case reduction from a problem about polynomial extrapolation to a problem about polynomial \emph{coefficient extraction}.
We consider a one-parameter family of matrices
\[A(t) = R+tW_{\text{dilute}}\]
and consider the case that $W_{\text{dilute}}$ consists of a tiny $n^\varepsilon$-sized worst case matrix $W'$ in direct sum with the all $0$'s matrix on the remaining $n-n^\varepsilon$ dimensions.
The key point of this construction is, even though the value of $|\Per(A(1))|^2$ is not what we want (as $A(1)=R+W_{\text{dilute}}$), the coefficients of the polynomial $|\Per(A(t))|^2$ do encode information about $\Per(W')$. In particular, the degree of the polynomial $|\Per(A(t))|^2$ is now $n^{2\varepsilon}$, and the top coefficient is $|\Per(W')|^2|\Per(R_D)|^2$, where $R_D$ is the bottom righthand minor of $R$ of dimension $n-n^\varepsilon$ (see Figure \ref{fig:dilution}):
\[|\Per(A(t))|^2 = |\Per(W')|^2|\Per(R_D)|^2 t^{2n^\varepsilon} + \displaystyle\sum_{\ell=0}^{n^{2\varepsilon}-1} c_\ell t^{\ell} \]
where the $c_\ell$ are some other coefficients which depend (in some complicated manner) on the entries of $R$ and $W$.
To see this, simply note that any term in the permanent which picks up all possible factors of $t$ must take all of its entries in the first $n^\varepsilon$ rows from the upper left submatrix.

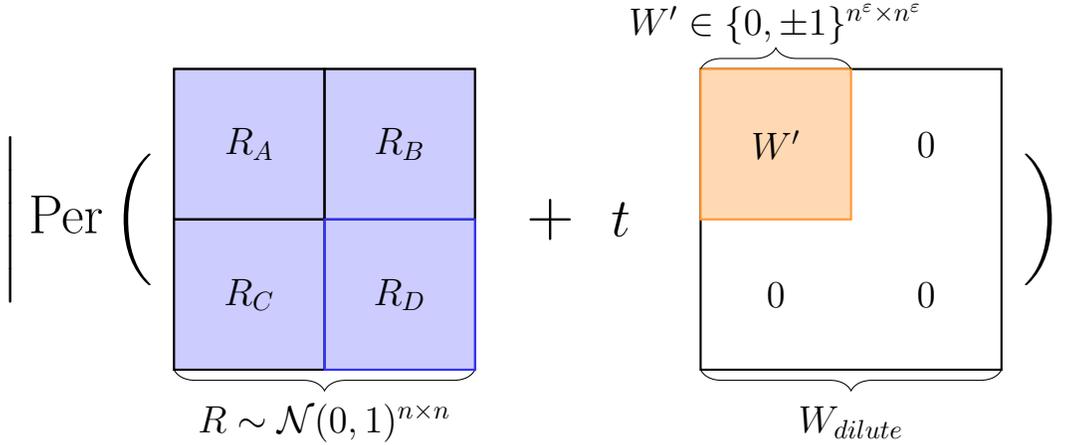
\begin{figure}[t!]

    \centering
    \begin{tikzpicture}[scale=2]

        \node[anchor=east] at (0,1) {\huge $\Bigg| \Per \bigg($};  
        \hspace{0.15cm}

        \draw[thick] (0,0) rectangle (2,2);
        \fill[blue!20] (0,0) rectangle (2,2);

        \node at (1,-0.35) {\textbf{\Large $R \sim \mcal N(0,1)^{n \times n}$}};
        \draw[decorate, decoration={brace, amplitude=8pt, mirror}] (0,0) -- (2,0) node[midway,xshift=0.4cm, yshift=-5pt] {};

        \node at (2.5,1) {\Huge $+$};

        \node at (2.9,1) {\huge \textbf{ $t$}};

        \draw[thick] (3.5,0) rectangle (5.5,2);

        \node at (4,0.5) {\textbf{\Large $0$}};
        \node at (5,0.5) {\textbf{\Large $0$}};
        \node at (5,1.5) {\textbf{\Large $0$}};

        \draw[thick] (0,0) rectangle (1,1);
        \draw[thick] (0,1) rectangle (1,2);
        \draw[thick] (1,1) rectangle (2,2);
        \draw[thick, blue!80] (1,0) rectangle (2,1);

        \node at (0.5,0.5) {\Large $R_{C}$};
        \node at (1.5,0.5) {\Large $R_{D}$};
        \node at (0.5,1.5) {\Large $R_{A}$};
        \node at (1.5,1.5) {\Large $R_{B}$};

        \fill[orange!30] (3.5,1) rectangle (4.5,2);  
        \draw[thick, orange!80] (3.5,1) rectangle (4.5,2);  

        \node at (4,1.5) {\textbf{\Large $W'$}};

        \draw[decorate, decoration={brace, amplitude=8pt}] (3.5,2) -- (4.5,2) node[midway,xshift=0.4cm, yshift=5pt] {};

        \node at (4,2.3) {\textbf{\Large $W' \in \{0,\pm 1\}^{n^\varepsilon \times n^\varepsilon}$}};

        \draw[decorate, decoration={brace, amplitude=8pt, mirror}] (3.5,0) -- (5.5,0) node[midway,xshift=0.4cm, yshift=-5pt] {};

        \node at (4.5, -0.35) {\textbf{\Large $W_{{dilute}}$}};

        \hspace{0.15cm}  
        \node[anchor=west] at (5.5,1) {\huge $\bigg) \Bigg|$};  

    \end{tikzpicture}
    \caption{In Theorem \ref{thm:dilution-intro}, we extract the coefficient of the polynomial $\vert\Per(R+tW_{dilute})\vert,$ where $R$ is a matrix of standard normals and $W_{dilute}$ has a worst-case matrix in its upper left block of size $n^\varepsilon\times n^\varepsilon$ for any constant $\varepsilon>0,$ with all other matrix entries being 0. The top coefficient of this polynomial is $\left|\Per W'\vert\vert \Per R_{D}\right|,$ where $R_D$ is the complementary minor to $W'.$}
    \label{fig:dilution}
\end{figure}

With this insight in hand, we can now give a new worst-to-average-case reduction for the permanent based on coefficient extraction: to compute $|\Per(W')|^2$ for some worst-case matrix $W' \in \{0,\pm1\}^{n^\varepsilon \times n^\varepsilon}$, pick many small values of $t$ ($t=O(1/n^\varepsilon)$ suffices by prior arguments) and compute $|\Per(A(t))|^2$ using our average-case algorithm. Then ask the $\NP$ oracle to give us a polynomial of degree $2n^\varepsilon$ which approximately matches these values.
Now look at the top coefficient of that polynomial, and divide by the value of $|\Per(R_D)|^2$.
Crucially, we can estimate the value of $|\Per(R_D)|^2$ to small multiplicative error, as this is another average-case instance.
As multiplicative error only adds under division, this now gives us a \emph{multiplicative} estimate for $|\Per(W')|^2$.
In other words, our algorithm translates relative error in the average case to relative error in the worst case.

We show that the overall robustness of this algorithm is merely $O(n^\delta)$ far in the exponent from showing quantum advantage, for any $\delta>0$ (Theorem \ref{thm:dilution-intro}).
The key point is that our polynomial now has degree $2n^\varepsilon$ rather than $2n$, and as such polynomial coefficient extraction incurs exponentially less error blowup. As noted earlier, our proof crosses the Jerrum-Sinclair-Vigoda barrier as this argument intrinsically requires that $W$ have mixed signs. 

The corollary for RCS follows by a similar dilution argument---one simply picks a worst case random circuit which is a concatenation of an $n^\varepsilon$ qubit worst case instance with an $(n-n^\varepsilon)$-sized random instance, and applies prior worst-to-average-case reductions \cite{movassagh2023hardness,Bouland2021,Kondo2021_robustness}.
See Appendix \ref{sec:dil-rcs} for details.

\subsubsection{Overcoming the convexity barrier: square method and magnification lemma}

While this first result exponentially improves on prior work, it is natural to ask how much closer we are to proving the Permanent-of-Gaussians Conjecture, or more generally to establishing hardness of sampling.
The above results are obtained by diluting the worst case instance size so as to lessen the error incurred by coefficient extraction.
However, the amount of error blowup relative to the worst case instance size has not improved.
At a deeper level, despite crossing the Jerrum-Sinclair-Vigoda barrier, the proof still does not imply hardness of sampling from Stockmeyer counting.
This is because Stockmeyer counting gives a $\BPP^\NP$ algorithm for approximating these squared permanents to $1/\poly(n)$ multiplicative error, but the worst-to-average-case reduction then blows up this error exponentially.
There is no compensating factor in the reduction to ``fight against'' this exponential loss. In other words, we have not yet crossed the convexity barrier. 

In our next set of results, we extend the coefficient extraction technique to cross the convexity barrier.
In particular we 
prove a new worst-to-average-case reduction for the permanent that can tolerate exponential losses from coefficient extraction, by developing two new techniques: the \emph{``square method''} and \emph{``magnification.''}

To do this, it is helpful to take a step back to examine what happens with \emph{dense} worst case matrices with our new coefficient extraction approach.
We apply two new modifications to coefficient extraction which improve the robustness of the dense case from $e^{-4n\log n -O(n)}$ \cite{krovicomm} to $e^{-1.5n\log n - O(n)}$.
While these modifications appear simple at first glance, we will see they introduce a term which we can use to combat the error blowup from coefficient extraction.
This dense result may at first look like a step backwards, but we will later show that this result is strong enough to imply a hardness of sampling theorem.

The first idea to improve robustness in the dense case, which we call the \emph{``square method,''} is to simply use the fact that $|\Per(A(t))|^2$ is the square of a polynomial to reduce the degree of coefficient extraction. 
Suppose our worst-case matrix $W$ is dense and define $A(t) \coloneq R+tW$ as before. 
While $|\Per(A(t))|^2$ is a degree $2n$ polynomial, trivially we have that
\[|\Per(A(t))|^2 = p(t)^2\] for some degree-$n$ polynomial $p(t)$.
In our reduction, after (approximately) computing $p(t)^2$ at many values of $t$ using our average-case algorithm, we can ask the $\NP$ oracle to give us the underlying degree $n$ polynomial $p(t)$ which squares to the correct value (up to the error tolerance in the average-case computation).
For real-value matrices, $p(t)$ is real, so is uniquely defined up to a sign.
Again the highest coefficient of this polynomial (now the coefficient of $t^n$) contains the value of $\Per(W)$ that we wish to compute.

One might a priori guess this simple change merely reduces the effective polynomial degree from $2n$ to $n$.
Surprisingly, it has more benefit than that!
In particular, suppose our average-case algorithm computes $p(t)^2$ to additive error $\pm \gamma$ at the points $t$ near $0$. How much error is induced on $p(t)$ itself? It turns out, $p(t)$ is estimated to \emph{less error} than $\gamma$.
Suppose our $\NP$ oracle gives us a polynomial $\tilde{p}(t) = p(t) + e(t)$ where $e(t)$ is some error polynomial.
Then trivially we have
\[ p(t)^2 \pm \gamma = (p(t)+e(t))^2 = p(t)^2 + 2 p(t)e(t) + e(t)^2\]
As our errors are vanishingly small in relative terms, the cross error term dominates, and we see
\[|e(t)| \leq \frac{\gamma}{p(t)}\]
at points $t$ near $0$.
In other words, we get to divide our error by the average-case value of the permanent, before we propagate the error through coefficient extraction. 
By assuming the Permanent Anticoncentration Conjecture \ref{conj:pacc-aa}, this value is $\sqrt{n!}$ to leading order, saving us an additional $\exp(\frac{n\log n}{2})$ beyond what we might have otherwise expected to gain in additive terms.
This observation gets more interesting if we view it in relative terms. 
This correction factor can be seen as ensuring the relative error on $p(t)$ is the same (up to a constant factor of 2) as the relative error on $p(t)^2$, as relative error is preserved (up to constants) under taking powers.

Observe that degree reduction via the square method kept our error constant in relative terms on our underlying polynomial. On the other hand, polynomial coefficient extraction is naturally sensitive to error in additive terms. 
Our next observation is that we can use this mismatch to \emph{reduce} the coefficient extraction error blowup in relative terms, by an exponential amount.
The basic idea is to now consider a worst case matrix with two components: first, a smaller and possibly negative-entry matrix $W'$ in the upper left hand corner of size $n^\varepsilon$, in direct sum with a larger matrix of all $1$'s of dimension $n-n^\varepsilon$ (see Figure \ref{fig:all-tricks}).

\begin{figure}[t!]
    \centering

    \definecolor{MediumGreen}{RGB}{0, 128, 0}
    
    \begin{tikzpicture}[scale=2]

        \node[anchor=east] at (0,1) {\huge $\Bigg| \Per \bigg($};  
        \hspace{0.15cm}

        \draw[thick] (0,0) rectangle (2,2);
        \fill[blue!20] (0,0) rectangle (2,2);
        \node at (1,1) {\textbf{\Large $R \sim \mcal N(0,1)^{n \times n}$}};

        \node at (2.5,1) {\Huge $+$};

        \node at (2.9,1) {\huge \textbf{ $t$}};

        \draw[thick] (3.5,0) rectangle (5.5,2);

        \fill[orange!30] (3.5,1.5) rectangle (4,2); 
        \draw[thick, orange!80] (3.5,1.5) rectangle (4,2);  

        \fill[MediumGreen!30] (4,0) rectangle (5.5,1.5);  
        \draw[thick, MediumGreen!80] (4,0) rectangle (5.5,1.5);

        \node at (3.75,0.75) {\textbf{\Large $0$}};
        \node at (4.75,1.75) {\textbf{\Large $0$}};

    \node[anchor=center] at (4.75, 0.75) {\large
        \begin{tabular}{ccccc}
            1 & 1 & 1 & 1 & 1 \\
            1 & 1 & 1 & 1 & 1 \\
            1 & 1 & 1 & 1 & 1 \\
            1 & 1 & 1 & 1 & 1 \\
            1 & 1 & 1 & 1 & 1 \\
        \end{tabular}
    };

        \draw[decorate, decoration={brace, amplitude=8pt}] (3.5,2) -- (4,2) node[midway,xshift=0.4cm, yshift=5pt] {};

        \node at (4,2.3) {\textbf{\Large $W' \in \{0,\pm1\}^{n^\varepsilon \times n^\varepsilon}$}};

        \draw[decorate, decoration={brace, amplitude=8pt, mirror}] (3.5,0) -- (5.5,0) node[midway,xshift=0.4cm, yshift=-5pt] {};

        \node at (4.5, -0.3) {\textbf{\Large $W$}};

        \hspace{0.15cm}  
        \node[anchor=west] at (5.5,1) {\huge $\bigg) \Bigg|$};  

    \end{tikzpicture}
    \caption{Polynomial $\vert\Per(R+tW)\vert,$ whose top coefficient is $\vert\Per W'\vert (n-n^\varepsilon)!.$ This is the ensemble under consideration in Thm.\ \ref{thm:hardness-of-probs} where we coefficient-extract the unsquared permanent via the square method and use worst case magnification by padding $W$ with a matrix of $1$s.}
    \label{fig:all-tricks}
\end{figure}
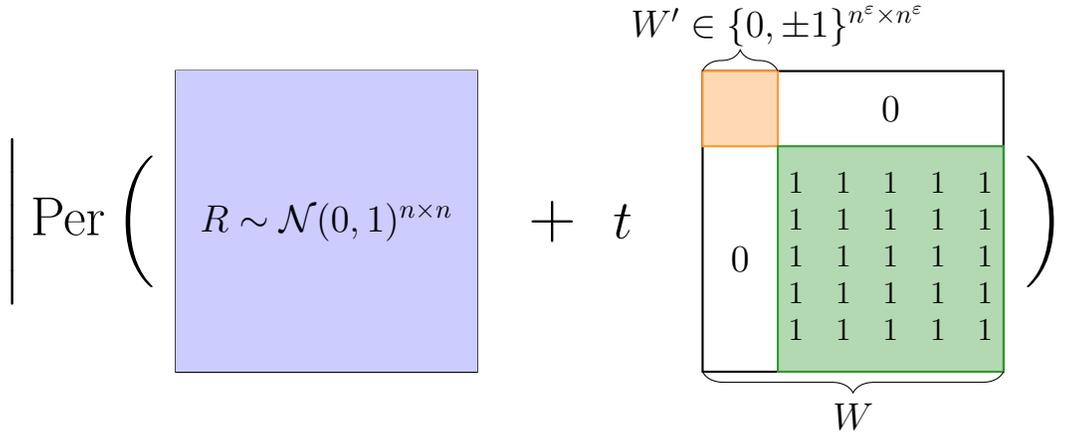

Interestingly, including this large-permanent submatrix in our worst case actually \emph{improves} our robustness in the worst-to-average-case reduction!
This is because for this scheme, the top coefficient of the polynomial $p(t)=\Per(A(t))$ is equal to $\Per(W')(n-n^\varepsilon)!$, where this \emph{magnification factor}, $(n-n^\varepsilon)!$, comes from the value of the permanent of the bottom right hand submatrix. 
Therefore, to compute $\Per(W')$ it suffices to estimate this top coefficient to additive error
$\frac{1}{3}(n-n^\varepsilon)!$
because $\Per(W')$ is integer-valued, so this error is removed by rounding to the nearest integer multiple of $(n-n^\varepsilon)!$.
In other words, the fact that this permanent of the all $1$s submatrix is big allows for more error tolerance in the reduction, overall improving the robustness.
We show this trick can be generalized to the more general formula:
\begin{restatable}[Magnification of robustness in worst-to-average-case reductions]{lemma}{magnification}
    It is $\sharpP$-hard to compute random Gaussian permanents on average to within relative error 
\[\gamma_{rel} \leq \frac{(n-n^\varepsilon)!}{\left|\Per R\right|}\cdot \Delta^n \cdot 2^{-O(n)}\]
for any constant $\varepsilon>0$, where $R\sim\mcal N(0,1)^{n\times n}$, 
and where $\Delta$ is the ``box size'' as in Sec.\ \ref{subsubsec:what-controls-robustness}, assuming a slight generalization of Permanent Anticoncentration, Conjecture \ref{conj:anticoncentration}.
    \label{lem:magnification}
\end{restatable}

In particular this worst-to-average-case reduction now has an exponential term---namely the ratio of the magnification factor to the average-case permanent---fighting against the exponential error blowup of polynomial coefficient extraction. 
For BosonSampling, this ratio is roughly $n!/\sqrt{n!} \approx \exp(\frac{n\log n}{2})$ which fights against a coefficient extraction error of $e^{-n\log n-O(n)}$,  resulting in a net relative error of $\approx\exp(-\frac{n\log n}{2})$ (to leading order) needed in the average case to show hardness of sampling. 
To show hardness of sampling in the average case, this means we ``merely'' need to reduce the exponential loss of coefficient extraction to a weaker exponential, or increase the value of the worst-case matrix (now all $1$s) by an exponential factor.
This is not an easy problem---these terms are interrelated, so say simply boosting the norm of the all $1$s matrix simultaneously improves the magnification-to-average-case ratio and worsens the coefficient extraction error, and does not show hardness of sampling.
However, we now finally have a term fighting \emph{against} coefficient extraction loss.
We note a similar lemma can be shown for RCS as well---in particular for a real version of RCS with random orthogonal gates (see Section \ref{sec:relative-error})---but does not yield any hardness of sampling results (see Discussion \ref{sec:discussion}).

\begin{figure}[t!]
    \centering
    \begin{tikzpicture}
        \draw[thick] (-4.5,0) -- (4.5,0);
        \foreach \x in {-3,-2,-1,0,1,2,3} {
            \draw[thick] (\x,0.2) -- (\x,-0.2);
            \node[below] at (\x,-0.2) {\small $\x$};
        }
        \node[above] at (0,0.5) {\small $\Per W' \in \mathbb{Z}$};
        \node[left] at (-4.5,0) {\small $\dots$};
        \node[right] at (4.5,0) {\small $\dots$};
        
        \draw[thick] (-7,-3) -- (7,-3);
        \foreach \x in {-3,-2,-1,0,1,2,3} {
            \draw[thick] (2*\x,-2.8) -- (2*\x,-3.2);
        }
        
        \node[below] at (-6,-3.2) {\small $-3(n-n^{\varepsilon})!$};
        \node[below] at (-4,-3.2) {\small $-2(n-n^{\varepsilon})!$};
        \node[below] at (-2,-3.2) {\small $-(n-n^{\varepsilon})!$};
        \node[below] at (0,-3.2) {\small $0$};
        \node[below] at (2,-3.2) {\small $(n-n^{\varepsilon})!$};
        \node[below] at (4,-3.2) {\small $2(n-n^{\varepsilon})!$};
        \node[below] at (6,-3.2) {\small $3(n-n^{\varepsilon})!$};
        
        \node[above] at (0,-2.5) {\small $\Per W = (n-n^{\varepsilon})! \Per W'$};
        \node[left] at (-7,-3) {\small $\dots$};
        \node[right] at (7,-3) {\small $\dots$};
        
        \draw[thick,->] (-5.5,0) to[out=180,in=-180] (-5.5,-2.5);
        
        \node at (-7.5,-1.25) {\includegraphics[width=2cm]{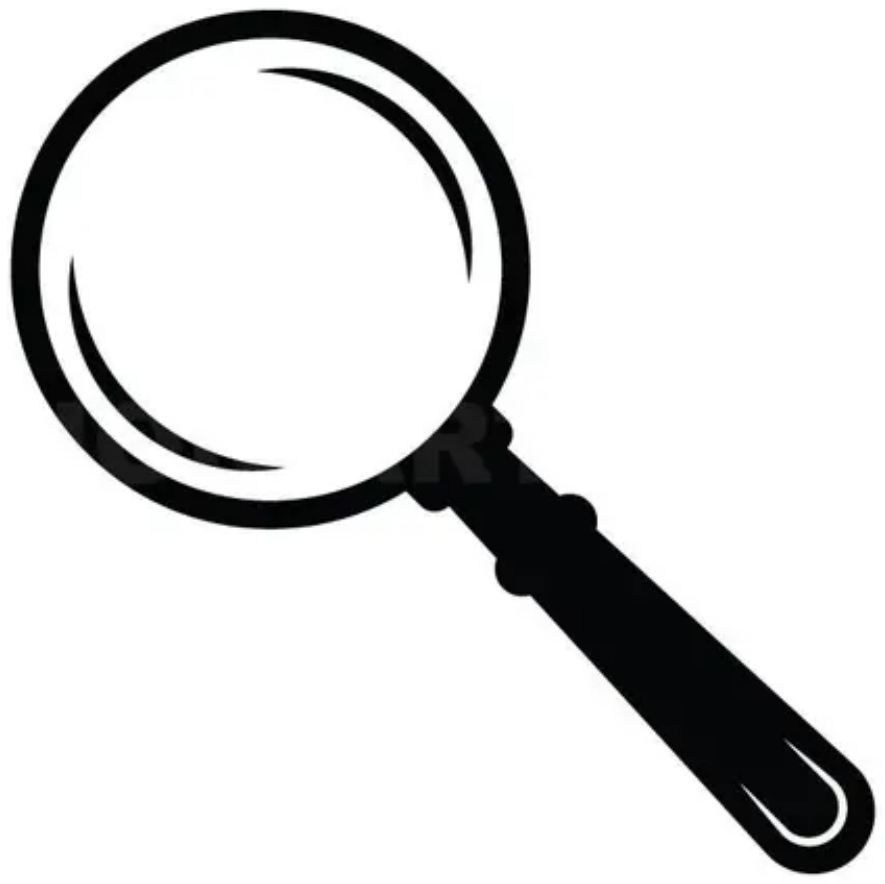}};
        
    \end{tikzpicture}
    \caption{Magnification: variation between the possible values of $\Per W',$ which is integer-valued, are \textit{magnified} by a factor of $(n-n^\varepsilon)!$ in $\Per W$ for any $\varepsilon>0$ we choose. We can instantiate this with e.g.\ the matrix $W$ in Fig.\ \ref{fig:all-tricks}. Since the gradations in $\Per W$ have been made so much wider, we can sustain more error while still computing $\Per W$ precisely (see Lemma \ref{lem:magnification}). 
    }
    \label{fig:mag}
\end{figure}

\subsubsection{Average-case hardness of sampling using random rare events lemmas}
Finally, we apply this new worst-to-average-case reduction to obtain the first nontrivial hardness of average-case sampling for BosonSampling.
This uses techniques specific to BosonSampling, which to the best of our knowledge do not carry over to other quantum advantage schemes.

To show this, we consider our new worst-to-average-case reduction, whose relative error robustness is given by Lemma \ref{lem:magnification}. To show an average-case hardness of sampling result via Stockmeyer, we need our relative error tolerance for $\#\mathsf{P}$-hardness to be inverse polynomial.
Our compensating ratio of the magnification factor to the average-case permanent is $\exp(\frac{n\log n}{2})$, so we can only afford this much error from coefficient extraction.
Unfortunately this is not enough of a loss budget to do a standard worst-to-average-case reduction.
This is because in these reductions, we compute values of $|\Per(A(t))|^2$ for values of $t$ which are small enough so that $A(t)$ is distributed close in total variation distance to Gaussian, to ensure our average-case algorithm correctly computes $A(t)$ with high probability.
To ensure closeness of total variation distance to constant error, $t$ must be $O(1/n)$---this calculation (due to Krovi \cite{krovicomm}) is optimal.
Recalling from Sec.\ \ref{subsubsec:what-controls-robustness} the discrete Remez inequality, Lemma \ref{lem:discrete-remez}, this yields an error blowup of $\sim n^n =e^{n\log n}$.
There is no hope of closing this gap with a standard total variation distance approach.

To get around this issue, our key idea is to go \emph{out of distribution}.
That is, what if we query points $A(t)$ which are \emph{far} from Gaussian distributed?
Clearly if our average-case algorithm could successfully compute the permanent of these matrices, then this would improve our robustness, as it would allow us to query points at much larger values of $t$, and hence reduce our error blowup.
For example, if we could successfully compute $|\Per(A(t))|^2$ for points $t=O(1/\sqrt{n})$, our coefficient extraction error would be halved in the exponent, and we could show hardness of average-case sampling!
However, the issue is these matrices $A(t)$ at large values of $t$ are far in total variation distance from Gaussian, so there is no trivial guarantee our algorithm works here.
In fact total variation distance arguments are useless here; the TV distance between $A(t)$ and Gaussian is of the form $1-\delta$ for a small value of $\delta$. Even if we assume our average-case algorithm works perfectly, a TV distance argument would only say it must work with probability at least $\delta$ on these points.
This is insufficient for our polynomial coefficient extraction techniques.

\begin{figure}[t!]
    \centering
    \begin{tikzpicture}[scale=1.85] 
        \def\shift{0.8} 
        \def\scalefactor{1.2}

        \draw[->] (-4,0) -- (4.5,0) node[right] {$t$};

        \draw[thick, black, domain=-4:4, samples=100] 
            plot (\x, {exp(-0.5 * (\x + \shift)^2) * \scalefactor});

        \draw[thick, black, domain=-4:4, samples=100] 
            plot (\x, {exp(-0.5 * (\x - \shift)^2) * \scalefactor});

        \draw[dashed] (1.5, 0) -- (1.5, 1.2);

        \begin{scope}
            \clip (-4,0) rectangle (1.5,1.2); 
            \fill[blue, opacity=0.2] 
                plot[domain=-4:1.5, samples=100] (\x, {exp(-0.5 * (\x + \shift)^2) * \scalefactor}) --
                (1.5,0) -- (-4,0) -- cycle;

            \fill[blue, opacity=0.2] 
                plot[domain=-4:1.5, samples=100] (\x, {exp(-0.5 * (\x - \shift)^2) * \scalefactor}) -- 
                (1.5,0) -- (-4,0) -- cycle;
        \end{scope}

        \begin{scope}
            \clip (1.5,0) rectangle (4,1.2); 
            \fill[purple, opacity=0.3] 
                plot[domain=1.5:4, samples=100] (\x, {exp(-0.5 * (\x + \shift)^2) * \scalefactor}) -- 
                (4,0) -- (1.5,0) -- cycle;
        \end{scope}

        \begin{scope}
            \clip (1.5,0) rectangle (4,1.2); 
            \fill[orange, opacity=0.3] 
                plot[domain=1.5:4, samples=100] (\x, {exp(-0.5 * (\x - \shift)^2) * \scalefactor}) -- 
                (4,0) -- (1.5,0) -- cycle;
        \end{scope}
        
    \end{tikzpicture}
    \caption{``Rare events'' lemma \ref{lem:probability-attenutation} shows that a function that computes permanents of $\mcal N(0,1)^{n\times n}$ matrices with $1-\exp(-O(n))$ probability also computes permanents of $\mcal N(t,1)^{n\times n}$ matrices with $1-1/\poly(n)$ probability for $t=O(1/\sqrt{n})$. That is, an algorithm that works very often over a Gaussian distribution will also work reasonably often on a shifted Gaussian distribution. The figure depicts that events deep in the tail of one Gaussian are still tail events for a shifted Gaussian, with successful events colored blue and failure events colored orange.} 
    \label{fig:shifted-gaussians}
\end{figure}
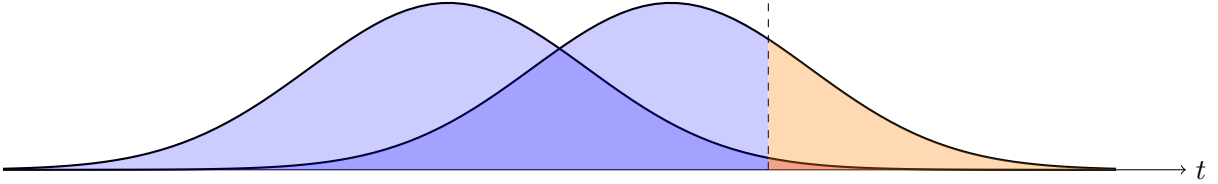

Instead, in our proof we go beyond total variation distance analysis to show that we can successfully query points $A(t)$ at high values of $t$, so long as our average-case algorithm works with very high probability.
The basic idea is this: suppose our average-case algorithm works near perfectly, say with probability $1-\delta$ over the choice of Gaussian matrix.
We want to show it also works if we query it on these points $A(t)$ which are far from Gaussian.
A basic observation is that these $A(t)$ are also Gaussian distributed, but with a shifted mean.
We prove a simple lemma, Lemma \ref{lem:probability-attenutation}, showing that rare events under one Gaussian distribution remain rare under another Gaussian, so long as their probability is less than $e^{-d^2}$ where $d$ is the distance between the means.
Intuitively this is because if an event is extremely far from the mean of a Gaussian $G_1$ (much further than the distance to the mean of $G_2$) it is also far from the mean of $G_2$ as well, and hence rare under $G_2$ (see Figure \ref{fig:shifted-gaussians}).
We then apply this lemma to the event that the average-case algorithm fails under the standard Gaussian.
If this is sufficiently rare for the average case, this is also rare for the distribution of $A(t)$, and hence the algorithm works with high probability to compute $A(t)$ as well.
There is a loss in this argument which forces $\delta$ to be exponentially small. 
However, the key point is that if our average-case algorithm works with extremely high probability, then it can also evaluate these points $A(t)$ at high values of $t$, and hence lessen the coefficient extraction error in our reduction.

We show this can be leveraged to show a nontrivial hardness of sampling result for an exact (i.e.\ multiplicative-error) average-case sampler, following the outline above---but the proof requires several additional technical innovations.
First, an average-case sampler that works with very high probability $1-\delta$ over the choice of BosonSampling experiment \emph{does not} immediately imply (by Stockmeyer counting) a $\BPP^\NP$ algorithm for computing Gaussian permanents with probability $1-\delta$. The issue is that submatrices of Haar random orthogonal matrices are not known to be exponentially close to Gaussian in TV distance, but rather have only been shown to be inverse polynomially close \cite{jiang-ma}. Thus setting the sampler success probability to $1-\delta$ where $\delta =2^{-O(n)}$ does not automatically yield a correspondingly good algorithm for computing Gaussian permanents.

To fix this we prove yet another ``rare events lemma,'' Proposition \ref{prp:small-small}, that allows us to transfer our algorithm for Haar submatrices to Gaussian matrices. The proof, which may be of independent interest, requires showing new results in random matrix theory, exploiting properties of the probability densities and spectra of i.i.d.\ Gaussian matrices and submatrices of Haar orthogonals.

Second, for our algorithm to work we require $\Per(A(t))$ to anticoncentrate.
This is not guaranteed by the standard Permanent Anticoncentration Conjecture \ref{conj:pacc-aa} as these matrices are out of distribution.
We instead formulate a more general conjecture that shifted mean Gaussian permanents anticoncentrate:

\begin{restatable}[Anticoncentration of gently perturbed Gaussian permanents]{conjecture}{anticoncentration}
\label{conj:anticoncentration}
    There exists a polynomial $f$ such that for all $n$ and $\epsilon>0,$
    \[
    \Prob_{R\sim\mcal N(0,1)^{n\times n} } \left[\left|\Per(R+tW)\right| < \frac{\sqrt{n!}}{f(n,1/\epsilon)}\right]<\epsilon,
    \]
    for arbitrary matrix $W$ with entries bounded by $1$ and $t=O(\frac{1}{\sqrt{n}}).$
\end{restatable}
We provide numerical evidence in support of Conjecture \ref{conj:anticoncentration} in Appendix \ref{sec:num}. Moreover, we note a special case of this conjecture has already been proven for $\mcal N(1/\poly\log n,1)$ matrices by \cite{ji2021approximating}, improving on work of Eldar and Mehraban \cite{eldar-mehraban}---and $N(0,1)$ matrices are the subject of standard anticoncentration---so our conjecture is in some sense interpolating between these proven statements and conjectures to matrices with entries like $\mcal N(1/\sqrt{n},1)$.
See Figure \ref{fig:anticoncentration} for a schematic.

\subsection{Discussion and open problems}\label{sec:discussion}

In this work we have exponentially improved over the best-known hardness results for BosonSampling, proving a robust worst-to-average-case reduction and showing the first non-trivial average-case multiplicative-error sampling result for (orthogonal) BosonSampling.
It is natural to ask if our techniques can be pushed further to prove PGC and show hardness of BosonSampling in the general case.
We note that further reductions in our coefficient extraction error could possibly yield intermediate results in this direction, in particular improving our average-case success probability of the sampler to be closer to $1-1/\poly(n)$. 
In terms of pushing our results towards approximate average-case approximate sampling (i.e., from a distribution close in total variation distance), an important question is if our techniques relativize, as we know non-relativizing techniques will be required to show hardness of approximate sampling \cite{Aaronson2017foundations}.
Interestingly Marshall, Aaronson and Djunko \cite{marshall2024improved} recently introduced new techniques that do not relativize.
Of course the Permanent Anticoncentration Conjecture \ref{conj:pacc-aa} remains open as well, and is assumed in our work.

Another natural question is if we can show any hardness of sampling for RCS similar to Theorem \ref{thm:no-sampler}. Here the principal challenge is that the state-of-the-art of average-case hardness for RCS is substantially farther from the goal than for BosonSampling \cite{bouland_complexity_2019,movassagh2023hardness,Bouland2021,Kondo2021_robustness,krovi2022}.
While we show one can utilize the schemes of Lemma \ref{lem:magnification} for a real variant of RCS to obtain a magnification-to-average-case ratio which fights against extrapolation loss (see Sec.\ \ref{sec:relative-error}), this gain is at most $2^n$ for RCS, while existing worst-to-average-case reductions have much larger robustness losses.
We leave this as an open problem.

We note a number of related works have studied the complexity of quantum advantage schemes under various forms of noise in the experiment, e.g.,~\cite{aharonov1996limitations,kalai2014gaussian,gao2018efficient,Bouland2021,aharonov2023polynomial,deshpande2022tight, schuster2024polynomial,oh2024classical,dalzell2024random,nonunital,bulmer2022boundary,villalonga2021efficient} which can make the problems asymptotically easier in certain scenarios.
In contrast our work is studying the complexity of near-noiseless variants of BosonSampling or RCS.

Finally, it remains open if our proofs can be extended from real (i.e.\ orthogonal) BosonSampling to complex (i.e.\ unitary) BosonSampling. 
The part of our proof that breaks here is the statement that, if you have evaluations of the square of a polynomial $|p(t)|^2$, that you can infer the underlying polynomial up to phase. While this is trivial in the real case (the phase is $\pm 1$, which is trivially disambiguated in the proof), in the complex case it is open if this approximately defines $p(t)$ up to a complex phase, and this appears to be an open problem in complex analysis \cite{mathoverflow1}.
We explain this in more detail in Appendix \ref{sec:complex-squares}.

\section{Background}\label{sec:background}
In this section, we record some background used throughout the paper.

In this work we will work to show hardness of \emph{exact} (often called \emph{multiplicative})\footnote{We note that in the literature, it is common to use \textit{``exact''} and \textit{``multiplicative-error''} sampling synonymously, simply because the techniques that demonstrate exact hardness typically extend to multiplicative-error hardness automatically. We will also adopt this convention, using the two terms interchangeably.} sampling of \emph{average-case} BosonSampling.
In BosonSampling the input is a Haar-random $m\times m$ unitary matrix $U$, describing a linear optical inteferometer on $m$ modes, and a number $n$ of photons.
The goal is to output a sample of the probability distribution obtained by passing those $n$ photons through the interferometer $U$ and measuring in the photon number basis.
We will work towards showing a classical algorithm cannot perform this task on average over the choice of $U$.
We define an exact average-case sampler to be the following:

\begin{definition}[Multiplicative-error average-case sampler]\label{def:sampler}
A multiplicative-error average-case sampler for BosonSampling that succeeds with probability $1-\alpha$ is an efficient classical probabilistic algorithm that, given a random $m\times m$ matrix $U$ and  error $\eta\geq0$, 
outputs a sample $y$ from a distribution $\mcal P'_U$ in time polynomial in $n$ and $1/\eta$ such that, with probability at least $1-\alpha$ over the choice of $U$, 
\begin{equation}
(1-\eta) P_U(y) < P'_U(y) \leq  (1+\eta)P_U(y)
\end{equation}
for every possible output $y$,  where $\mathcal{P}_U$ is the output distribution of the BosonSampling experiment.
\end{definition}
We note that this notion of sampling to small multiplicative error has been previously studied in the literature, e.g. in \cite{terhal2004adptive,bremner2010iqp,Aaronson2013}, where it has been shown that \emph{worst-case} multiplicative-error sampling is hard assuming $\textsf{PH}$ does not collapse.
In this work we prove a similar statement for \emph{average-case} multiplicative-error sampling.

Aaronson and Arkhipov gave a well-known reduction from classical sampling to approximate computing of output probabilities that uses Stockmeyer's approximate counting algorithm \cite{stockmeyer1983complexity}, which runs in $\mathsf{BPP}^\mathsf{NP}$ \cite{Aaronson2013}. The idea of Stockmeyer's algorithm is to estimate the probability of any outcome by estimating the number of random strings that cause the sampler to output that outcome. This uses that a classical randomized algorithm can be treated as a deterministic algorithm that takes a random input.
They then use this to show an efficient classical sampler cannot exist.
The basic idea is that if approximately computing output probabilities is $\#\mathsf{P}$-hard, then it cannot lie in $\BPP^\NP$ by Toda's theorem \cite{Toda1991}.
Thus to show hardness of sampling, it suffices to conjecture that it is hard to compute the output probabilties of experiments.

Aaronson and Arkhipov showed that the Permanent-of-Gaussians Conjecture (PGC) suffices to prove hardness of average-case approximate sampling.
PGC states that the following problem is \textsf{\#P}-hard:
\begin{definition}[\textsf{GPE}$_\pm$]
    Given as input $R\sim\mcal N(0,1)^{n\times n}$ and error parameters $\epsilon, \delta > 0,$ estimate $\left|\Per R\right|^2$ to within additive error $\pm\epsilon\cdot n!$ with probability at least $1-\delta$ over $R,$ in $\poly(n, 1/\epsilon, 1/\delta)$ time.
\end{definition}
 This conjecture is natural in the context of BosonSampling as the output probabiltiies of these experiments correspond to matrix permanents of submatrices of the input $U$.
 For Haar-random unitaries of sufficient large dimension $m$, the submatrices are close to Gaussian, so this conjecture is
referring to the complexity of computing output probabilities of the experiment.

Aaronson and Arkhipov also assume the Permanent Anticoncentration Conjecture, which posits a lower bound on the typical value of these permanents. 

\begin{conjecture}[Permanent Anticoncentration Conjecture (PACC) \cite{Aaronson2013}]\label{conj:pacc-aa}
        There exists a polynomial $f$ such that for all $n$ and $\epsilon>0,$
    \[
    \Prob_{R\sim\mcal N(0,1)^{n\times n} } \left[\left|\Per R\right| < \frac{\sqrt{n!}}{f(n,1/\epsilon)}\right]<\epsilon.
    \]
\end{conjecture}

Assuming PACC, $\mathsf{GPE_\pm}$ and estimating Gaussian permanents to $1/\poly(n)$ relative error, a problem known as \textsf{GPE}$_{\times},$ are polynomial-time equivalent. Thus if Conjecture \ref{conj:pacc-aa} holds, then it suffices to show that \textsf{GPE}$_{\times}$ is \textsf{\#P}-hard.

Finally, we note that in standard BosonSampling $U$ is assumed to be a Haar random unitary matrix.
In this work all many of our results  (all except the coefficient extraction and dilution arguments of Theorem \ref{thm:dilution-intro}, which apply to unitaries) pertain only to the case that $U$ is a Haar random orthogonal matrix, whose $n\times n$ submatrices are distributed as $\mathcal{N}(0,1)^{n\times n},$ i.e.\ our matrix entries are real i.i.d.\ standard normals. This is because the square method, developed in Section \ref{sec:square-trick}, applies to real polynomials. In Appendix \ref{sec:complex-squares}, we explain why the generalization to complex polynomials relies on a difficult question in complex analysis. 
Another advantage of random $m\times m$ orthogonal linear optical transformations is that there is a proof that $n \times n$ submatrices are close in total variation distance to i.i.d.\ Gaussian matrices for $m=\Omega(n^2)$ \cite{jiang-ma}. It is widely conjectured in the BosonSampling literature that the same holds for random unitaries. Our focus on orthogonal matrices means that we sidestep this conjecture altogether.

\section{Techniques}
\input{top-coefficient.tex}

\input{real-squares.tex}

\subsection{Worst case magnification}\label{sec:relative-error}

Building on the square method in Sec.\ \ref{sec:square-trick}, in this section we develop the method of worst case magnification. 

\magnification*

\begin{remark}[Overcoming the convexity barrier of \cite{Aaronson2013}]
    The ``convexity'' barrier identified in \cite{Aaronson2013} observes that worst-to-average-case reductions based on numerical tasks such as polynomial extrapolation or coefficient extraction cannot exactly estimate a worst-case (i.e.,\ arbitrary) permanent given only $1/\poly(n)$ relative error on random instances in the average-case, simply because the reduction suffers from an exponential loss without a factor to compensate. Lemma \ref{lem:magnification} surmounts the convexity barrier by demonstrating that one can magnify the reduction's error tolerance by an exponentially large factor $(n-n^\varepsilon)!/|\Per R |,$ fighting against the exponential loss $\Delta^n 2^{-O(n)}= O(2^{-n\log n})$. 
    Box size $\Delta$ is explicitly calculated in Appendix \ref{sec:pinsker} to be $O(1/n)$.
\end{remark}

\begin{proof}
Let us in particular take $W$ to be a block-diagonal matrix consisting of a (small) block $W'\in\{0,\pm 1\}^{n^\varepsilon\times n^\varepsilon}$ for some constant $\varepsilon>0$, in direct sum with a $(n-n^\varepsilon)\times(n-n^\varepsilon)$ block of the all $1$s matrix, as illustrated in Fig.\ \ref{fig:all-tricks}. Observe that $\Per W'\in\mathbb Z$ and that the all $1$s matrix has a permanent of $(n-n^\varepsilon)!$. Furthermore, observe that the leading-order coefficient $q_n=\Per W = \Per W' \cdot (n-n^\varepsilon)!.$  Therefore see that the values of $q_n=\Per W$ are \emph{magnified}
by quantity $(n-n^\varepsilon)!,$ namely that $q_n$ can take on values $\ldots -2(n-n^\varepsilon), -(n-n^\varepsilon), 0, (n-n^\varepsilon), 2(n-n^\varepsilon)\ldots$.\footnote{ We note independent work of \cite{bravyi2025classical} used a related concept known as ``granularity'' in a quantum advantage setting. While superficially similar, these ideas are different, as in our setting we are \emph{artificially} making a \emph{worst-case} more ``granular'' to magnify our error tolerance, whereas in their setting the granularity naturally arises in their \emph{average-case} distribution from the integrality of representation-theoretic quantities.
} This is illustrated in Fig.\ \ref{fig:mag}.

This motivates the following worst-to-average-case reduction in $\mathsf{BPP}^{\mathsf {NP}}$: ask the \textsf{NP} oracle for a degree $n$ polynomial $p$ such that $\sup_{x\in S} |p^2(x) - q^2(x)| \leq \gamma$ for a set $S$ of $2n+1$ evenly-spaced points in $[-1/n,1/n].$ This has an efficient certificate because by assumption, we can evaluate average-case permanents, specifically $\Per(R+tW)$ for $t\in\Delta=O(1/n),$ and simply check that for all $x\in S,$ $p^2$ and $q^2$ are $\pm \gamma$-close. Then we may simply output the leading-order (degree $n$) coefficient $p_n$ of $p.$

The \emph{key idea} is that by magnifying $\Per W' \in \mathbb Z$ by a factor of $(n-n^\varepsilon)!$, it suffices for $|p_n|$ to be $\pm\frac{1}{3}(n-n^\varepsilon)!$ close to $|q_n|$ to compute $\Per W'$ exactly.\footnotemark~This would imply that computing $|q(0)|^2 = \vert \Per R\vert^2$ to within $\pm \gamma$ is $\sharpP$-hard under $\mathsf{BPP}^{\mathsf {NP}}$, i.e.\ that the average case is as hard as the worst case.

To complete the proof, we upper bound $\gamma$ by recalling Lemma \ref{lem:square-trick-coex} from the previous page: for real-valued polynomials $p$ and $q$ of degree $d$ and for a $\delta$-separated collection of points $S$ in $[-\ell,\ell]$ for which $|S|\geq 2d+1,$ if 

\[
\sup_{x\in S} |p^2(x) - q^2(x)| \leq \gamma
\]

and $$\inf_{x\in S} |q(x)| \geq K,$$

then 

\begin{equation}\label{eq:sq-coex-restatement} 
||p_d| - |q_d|| \leq 2^{2d+1} e^{2d} (d\delta)^{-d} K^{-1}\gamma,
\end{equation}
where $p_d$ and $q_d$ denote the leading-order coefficients of $p$ and $q.$\\

To apply this lemma, we instantiate $p,$ $q,$ $d,$ $\delta,$ $S,$ and $K$ as follows.
\begin{itemize}
    \item Take polynomial $q(t)=\Per (R+tW),$ where $R\sim\mcal N(0,1)^{n\times n}$ and $W$ is a matrix whose permanent is \textsf{\#P}-hard. This is a degree $d=n$ polynomial.
    \item $S$ is a finite set of points $\{t_i\}$ 
    in $[-\Delta,\Delta]$ for which box size $\Delta=O(1/n)$ and $|S|\geq 2n+1.$  Therefore, $S$ is a collection of $\delta=O(1/n^2)$ points, and the expression $(d\delta)^{-d}$ in Eq.\ \ref{eq:sq-coex-restatement} simplifies to $O(\Delta)^{-n} = \Delta^{-n} \cdot 2^{-O(n)}$. 
    \item $p$ is a polynomial of the same degree as $q,$ whose squared values are $\pm \gamma$-close to $q^2$ for $x\in S.$ 
    \item  Conjecture \ref{conj:anticoncentration} implies $\vert q(x)\vert =\vert \Per(R+xW)\vert \geq K=\sqrt{n!}/\poly(n)$ for all $x\in S$ with probability $1-1/\poly(n)$ by applying a union bound.
\end{itemize}

Using Eq.\ \ref{eq:sq-coex-restatement} and the bulleted substitutions above gives

\begin{align} 
||p_n| - |q_n|| &\leq 2(2e)^{2n} O(\Delta)^{-n} (\sqrt{n!}/\poly(n))^{-1}\gamma \\&\leq \frac{(n-n^\varepsilon)!}{3}.
\end{align}

Re-arranging, 

\begin{align} 
\gamma &\leq \frac{(n-n^\varepsilon)!}{3} \cdot \frac{\sqrt{n!}}{\poly(n)} \cdot O(\Delta)^n\cdot \frac{(2e)^{-2n}}{2} 
\\&=(n-n^\varepsilon)!\cdot \frac{\sqrt{n!}}{\poly(n)} \cdot \Delta^{n} \cdot 2^{-O(n)}.   \label{eq:gamma-add}
\end{align}

Next we ask, how well does the value of $\vert p(0)\vert^2$ approximate that of $\vert q(0)\vert^2=\vert\Per R\vert^2$? To answer this question, we define 

\begin{align}
    \gamma_{rel} \coloneqq \frac{\vert\vert p(0)\vert^2-\vert q(0)\vert^2\vert}{\vert q(0)\vert^2} 
    &= \frac{\vert\vert p(0)\vert^2-\vert \Per R\vert^2\vert}{\vert \Per R\vert^2} 
    \\&\leq \frac{\gamma}{\vert\Per R\vert^2},
\end{align}
recalling that $\gamma$ is the maximum additive error between $p^2$ and $q^2$ for $x\in [-\Delta,\Delta].$ 

Using Eq.\ \ref{eq:gamma-add}, we can suggestively bound $\gamma_{rel}$:

\begin{align}
     \gamma_{rel} \leq \frac{\gamma}{\left|\Per R\right|^2}
     &\leq 
     (n-n^\varepsilon)!\cdot \frac{\sqrt{n!}/\poly(n)}{\vert\Per R\vert^2} \cdot \Delta^{n} \cdot 2^{-O(n)}
     \\ &=\frac{(n-n^\varepsilon)!}{\vert\Per R\vert} \cdot \Delta^{n} \cdot 2^{-O(n)},
\end{align}
where in the final line we invoke Permanent Anticoncentration Conjecture \ref{conj:pacc-aa}.

\end{proof}

\footnotetext{The astute reader will notice that in this reduction we obtain $\Per W$ exactly, rather than to within some amount of relative error as we did in Theorem \ref{thm:dilution-intro}. Consequently, this argument does not cross the Jerrum-Sinclair-Vigoda (JSV) barrier. However, a simple modification to the proof crosses the convexity and JSV barriers simultaneously: Take as $W$ a block-diagonal matrix comprising the following three matrices in direct sum: $W'\in\{0,\pm 1\}^{n^\varepsilon\times n^\varepsilon},$ the all $0$s matrix of size $n^\varepsilon\times n^\varepsilon,$ and the all $1$s matrix of size $(n-2n^\varepsilon)\times (n-2n^\varepsilon).$ Just as in the proof of Theorem \ref{thm:dilution-intro}, the introduction of a random minor into the leading-order coefficient of $\Per(R+tW)$ makes it so that the reduction obtains a relative error approximation to $\Per W$. 
}

\subsection{Rare events lemma I: going out of distribution} 

Even with the exponential gains in error tolerance made using the square method and magnification, they still do not yet show the hardness of average-case sampling. The limitation comes from total variation distance analysis in the standard worst-to-average-case reduction, where we evaluate permanents drawn from a distribution close in TVD to i.i.d.\ Gaussian (see Sec.\ \ref{subsubsec:what-controls-robustness}).

Our next technical innovation is to go out of distribution. 
Namely, we prove what we call a ``rare events lemma," which show that an algorithm that computes Gaussian permanents with sufficiently high probability can also compute permanents distributed \textit{far} in total variation distance from i.i.d.\ Gaussian reasonably well. As these results exploit the specific structure of the Gaussian measure, they pertain only to BosonSampling.

\begin{lemma}\label{lem:shifted-gaussian-integral}
Let $S\subset\Real^N$ be a measurable set and let 
\[
\delta := (2\pi)^{-N/2} \int_S e^{-\|x\|^2/2} \diff x.
\]
Then for all $v\in\Real^N$,
\[
(2\pi)^{-N/2} \int_S e^{-\|x-v\|^2/2} \diff x \leq
e^{\|v\|^2/2} \delta^{1/2}.
\]
\end{lemma}
\begin{proof}
Let $\chi_S$ be the indicator function for the set $S$.  We compute
\begin{equation}
\begin{split}
(2\pi)^{-N/2}
\int_S e^{-\|x-v\|^2/2} \diff x
&= (2\pi)^{-N/2}\int e^{-\|x-v\|^2/2} \chi_S(x) \diff x \\
&= (2\pi)^{-N/2} e^{-\|v\|^2/2} \int e^{-\|x\|^2/2} e^{v\cdot x} \chi_S(x) \diff x \\
&\leq e^{-\|v\|^2/2} \Big((2\pi)^{-N/2}\int e^{-\|x\|^2/2} e^{2v\cdot x} \diff x\Big)^{1/2} \\
&\qquad \qquad \Big((2\pi)^{-N/2} \int e^{-\|x\|^2/2} \chi_S(x)\diff x\Big)^{1/2} \\
&= e^{-\|v\|^2/2} \Big(e^{2\|v\|^2} (2\pi)^{-N/2} \int e^{-\|x-2v\|^2/2} \diff x\Big)^{1/2} \\ 
&\qquad\qquad \Big((2\pi)^{-N/2} \int e^{-\|x\|^2/2} \chi_S(x)\diff x\Big)^{1/2} \\
&= e^{\|v\|^2/2} \delta^{1/2}.
\end{split}
\end{equation}
In the inequality above we applied Cauchy-Schwartz by
writing
\[
e^{-\|x\|^2/2} e^{v\cdot x} \chi_S(x)
= (e^{-\|x\|^2/4} e^{v\cdot x})
(e^{-\|x\|^2/4} \chi_S(x)).
\]
\end{proof}

Next we prove our first rare events lemma, Lemma \ref{lem:probability-attenutation}, depicted in Fig.~\ref{fig:shifted-gaussians}.
\begin{lemma}[Rare events lemma I]\label{lem:probability-attenutation}
Take $A\sim \mathcal N(0,1)^{n \times n}.$ Let $g:\Real^{n\times n} \to\Real$ be a function such that 
\[
||\Per(A)|^2 - g(A)|\leq \eps
\]
holds with probability $1-\delta$.  Let $B$ be an arbitrary matrix with entries $|b_{ij}|\leq 1$.  Then 
\[
||\Per(A+tB)|^2 - g(A+tB)| \leq \eps
\]
holds with probability at least $1-\sqrt{e^{\|tB\|^2}\cdot\delta}$, where $\|\cdot\|$ is the Hilbert-Schmidt norm.
\end{lemma}

\begin{proof}
Let $S\subset\Real^{n\times n}$ be the set 
\[
S := \{A \mid \Real^{n\times n} \mid 
||\Per(A)|^2 -g(A)| >\eps\}.
\]
We apply Lemma \ref{lem:shifted-gaussian-integral} with the set $S$
above on $\Real^N = \Real^{n\times n}$, where $\Prob(A\in S) = \delta$ and 
thus the lemma shows that 
$\Prob(A+tB\in S) \leq e^{\|tB\|^2/2} \delta^{1/2}$ as desired.
\end{proof}

We will leverage this lemma, in combination with the second rare events lemma in the next section, to prove the first nontrivial average-case harness of sampling theorem in Sec.\ \ref{sec:results-sampling}.

\subsection{Rare events lemma II: tail probabilities for orthogonal submatrices and i.i.d.\ Gaussians}

To prove our hardness of sampling Theorem \ref{thm:no-sampler}, further technical innovations are required beyond our first rare events Lemma \ref{lem:probability-attenutation}. A key issue that remains is that an average-case sampler that works with very high probability $1-\delta$ over the choice of BosonSampling experiment does \emph{not} immediately imply (by Stockmeyer counting) a $\BPP^\NP$ algorithm for computing Gaussian permanents with probability $1-\delta$. The issue is that submatrices of Haar random orthogonal matrices are not known to be exponentially close to Gaussian in TV distance, but rather have only been shown to be inverse polynomially close \cite{jiang-ma}. Thus setting the sampler success probability to $1-\delta$ where $\delta =2^{-O(n)}$ does not automatically yield a correspondingly good algorithm for computing Gaussian permanents.

To fix this, we prove yet another ``rare events'' lemma, Proposition \ref{prp:small-small}, that allows us to transfer our high probability algorithm for Haar submatrices to Gaussian matrices. The proof requires some highly nontrivial random matrix theory, exploiting properties of the probability densities and spectra of i.i.d.\ Gaussian matrices and submatrices of Haar orthogonals, and may be of independent interest.

Formally, we consider two models of $n\times n$ random matrices.  The first is a Gaussian matrix $X^n$ with
independent (real) entries of variance $n^{-1}$.  It has a probability density given
by
\[
p_G(X) = Z_G^{-1}(n) ( \prod_{i\in[n]} \exp(-n\lambda_i(X^TX)/2)
\]
where $\lambda_i(A)$ is the $i$-th eigenvalue of $A$.  The factor of $n$ comes from the normalization we apply, and $Z_G^{-1}(n)$ is
a normalization constant so 
that
\[
\int_{\Real^{n\times n}} p_G(X)\diff X = 1.
\]

The second model is that of a $n\times n$ submatrix of a Haar-random $m\times m$ orthogonal matrix.
We rescale by $\sqrt{m/n}$ so that the individual entries have variance $n^{-1}$.
Then for $m\geq 2n$ the probability density takes the form (see \cite{jiang-ma}, Lemma 2.1)
\[
p_S(X) = Z_S^{-1}(n,m) \prod_{i\in[n]} (1-n\lambda_i(X^TX)/m)^{(m-2n)/2} \One_{\lambda_i \leq m/n}.
\]

Our main result in this section is the following:

\begin{restatable}[Rare events lemma II]{proposition}{smallsmall}
\label{prp:small-small}
Let $E\subset\Real^{n\times n}$ be a measurable subset of
matrices, and suppose that
\[
\Prob_S(E) \leq \delta
\]
when $E$ is sampled as the $n\times n$ submatrix of a Haar-random $m\times m$ orthogonal matrix,
scaled by $\sqrt{m/n}$ so that each entry has variance $n^{-1}$.   Let $0<\alpha\leq1$ and
suppose that $n>C$ and $m>Cn^2$ for some absolute 
constant $C$.  Then
\[
\Prob_G(E) \leq 3\exp(-n^\alpha) + 10\exp( n^{\alpha/2})\delta,
\]
where $\Prob_G$ indicates that $E$ is sampled with independent Gaussian entries of
variance $n^{-1}$.
\end{restatable}

Proposition \ref{prp:small-small} has a highly nontrivial proof that we give in Appendix \ref{sec:rare-events}, using analytic and random matrix theory techniques.

\section{Hardness of computing output probabilities, Theorem \ref{thm:dilution-intro}}\label{sec:wtac}

Our first result makes progress towards proving the Permanent-of-Gaussians Conjecture (PGC).
Theorem \ref{thm:dilution-intro} gives a new worst-to-average-case reduction for computing Gaussian permanents whose additive error tolerance exponentially improves on the state-of-the-art. For the first time, our error tolerance matches to leading order that of the Permanent-of-Gaussians Conjecture (PGC), $\exp\big({-n\log{n}-n-O(n^\delta)}\big)$ compared to the goal of $\exp\big({-n\log{n}-n-O(\log n)}\big)$. All that remains is ``merely'' to improve the $O(n^\delta)$ term in the exponent to $O(\log n)$.

\dilutionrestatable*

\begin{proof} 
Take any arbitrary constants $\delta > \varepsilon>0$ and for ease of notation, define $k\coloneq \lfloor n^\varepsilon \rfloor.$

Recall that the output probability of a BosonSampling experiment is 
$$p_R\coloneq\frac{\left|\Per R\right|^2}{m^n} = \frac{\left|\Per R\right|^2}{n^{2n}},$$
where $R\sim\mcal N(0,1)^{n\times n}$ and the number of modes $m=\Theta(n^2).$ Let $\mcal A$
be an algorithm that given as input $R$ approximates $p_R$ up to additive error $\gamma$, with success probability at least $1-\eta$ over the choice of $R$ for some constant $\eta<1/4.$ Additionally, consider a ``worst-case'' matrix $W_{dilute}$ consisting of an upper-left block $W'\in\{0,\pm1\}^{k\times k}$ with all other entries being 0.

We will show that then there exists a $\mathsf{BPP}^{\mathsf{NP}^{\mathcal A}}$ procedure that given as input \emph{any} matrix $W_{dilute},$ approximates $\left|\Per W'\right|$ up to 
small relative error for 
$\gamma=\exp(-n\log n -n - O(n^\delta)),$
with constant success probability $1-\eta'$ for $\eta'$ slightly $> \eta$. 
The theorem statement follows immediately from the \textsf{\#P}-hardness of computing even a multiplicative approximation to the permanent of a $\{0,\pm 1\}$ matrix. 

Define the polynomial 
\begin{equation}\label{eq:dilution-poly}
    \left|\Per(A(t))\right| \coloneq \left|\Per(A(0)+tW_{dilute})\right|,
\end{equation}
where $A(0)\sim\mcal N(0,1)^{n\times n}$ and and $W_{dilute}$ is as above. Then $\left|\Per(A(t))\right|$ is a degree $k$ polynomial in $t$ whose leading coefficient is  $\left|\Per W'\right|\left|\Per R'\right|,$ where $R'$ is the complementary minor to $W'.$ This polynomial is illustrated in Fig.\ \ref{fig:dilution}.

As computed in Lemma \ref{lem:tvd-pinsker}, the total variation distance between the distributions of $A(t)$ and $A(0)$ is $O(kt).$ This follows from the KL divergence between two translated Gaussians and an application of Pinsker's inequality.

Consider $O(k)$ equally spaced points $\{t_i\}$ in the interval $[0,\Delta]$ for $\Delta=O(1/k)$. For suitable choice of constants, we can ensure that for each $t_i$,
\[
\Pr\left[ \bigg| \mcal A (A(t_i)) - \frac{\left|\Per(A(t_i))\right|^2}{n^{2n}}\bigg| \geq \gamma\right]\leq \eta + O(k \Delta) \leq \eta'
\]
for some slightly larger constant $\eta'.$ 
Then the $\mathsf{BPP}^{\mathsf{NP}^{\mathcal A}}$ procedure is as follows: query the \textsf{NP} oracle for a degree $k$ polynomial $q$ such that $|q(t_i)|^2/n^{2n}$ is $\pm \gamma$-close to the value obtained by $\mcal A$ for at least half of the points $\{t_i\}$. 
This admits a certificate that can be efficiently verified by checking each point $\{t_i\}$ for agreement between $\mcal A$ and $|q(t_i)|^2/n^{2n}.$ Return as output 
$|q_k|/\left|\Per R'\right|.$ 

Finally, we will use Lemma \ref{lem:square-trick-coex} to guarantee that additive error $\gamma = \exp(-n\log n -n - O(n^\delta))$ gives a good relative error estimate of $\left|\Per W'\right|$. Call $S$ the subset of points $\{t_i\}$ at which $\mcal A$ and $|q(t_i)|^2/n^{2n}$ agree and observe that the points are $O(1/k^2)$-separated. Moreover by permanent anticoncentration, $\inf_{t\in S} |q(t)| \geq \sqrt{n!}/\poly(n)$
with probability at least $1-1/\poly(n).$ Recalling that the leading-order coefficient of Eq.\ \ref{eq:dilution-poly} is $\left|\Per W'\right| \left|\Per R'\right|$ and that $\left|\Per R'\right|\geq\sqrt{n!}e^{-O(k\log n)},$ Lemma \ref{lem:square-trick-coex} gives that $\vert(\left|\Per W'\right| - \frac{\left|q_k\right|}{\left|\Per R'\right|})\vert \ll \vert\Per W'\vert$ if $\gamma/\left|\Per R\right|^2 = e^{-O(k\log n)}=e^{-O(n^\delta)}.$

Overall, we have a $\mathsf{BPP}^{\mathsf{NP}^{\mathcal A}}$ procedure to multiplicatively estimate $\left|\Per W'\right|$ if $\gamma = \exp(-n\log n -n - O(n^\delta))$, which concludes the proof.
\end{proof}

This proof technique carries over to Random Circuit Sampling, which we show in Appendix \ref{sec:dil-rcs}. 

\section{Hardness of sampling, Theorem \ref{thm:no-sampler}}\label{sec:results-sampling}

The goal of this section is to prove the following theorem, which closes the robustness gap for the first time at the expense of winnowing the failure probability to which we can prove hardness from $1/\poly(n)$ to $1/\exp(O(n))$.

\nosampler*

To prove Theorem \ref{thm:no-sampler}, we will assume the following anticoncentration conjecture, illustrated in Fig.~\ref{fig:anticoncentration}. We provide numerical evidence for Conjecture \ref{conj:anticoncentration} in Appendix \ref{sec:num}.

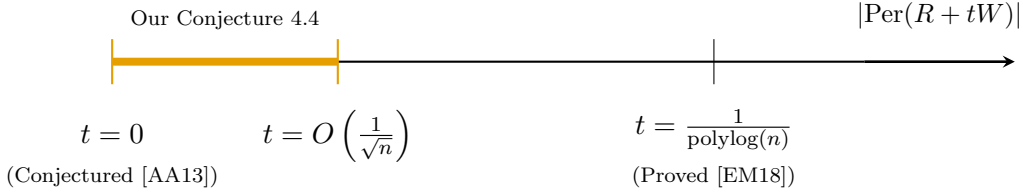
\begin{figure}[t!]
    \centering
    \begin{tikzpicture}

        \definecolor{accessibleorange}{RGB}{230, 159, 0}

        \draw[line width=3pt, accessibleorange] (0,0) -- (3,0);
        \draw[thick] (3,0) -- (10,0);
        
        \draw[thick,->, >=stealth, line width=1pt] (10,0) -- (12,0);

        \draw[thick, accessibleorange] (0, -0.3) -- (0, 0.3); 
        \draw[thick, accessibleorange] (3, -0.3) -- (3, 0.3); 
        \draw (8, -0.3) -- (8, 0.3);

        \node[below] at (0, -0.7) {$t = 0$};
        \node[below] at (3, -0.5) {$t = O\left(\frac{1}{\sqrt{n}}\right)$};
        \node[below] at (8, -0.5) {$t \geq \frac{1}{\text{polylog}(n)}$};

        \node[below] at (0, -1.25) {\scriptsize (Conjectured \cite{Aaronson2013})};
        \node[below] at (8, -1.25) {\scriptsize (Proved \cite{eldar-mehraban,ji2021approximating})};

        \node[above, font=\scriptsize] at (1.5, 0.3) {Our Conjecture \ref{conj:anticoncentration}};

        \node[above] at (11, 0.3) {\small $\left|\Per (R+tW)\right|$};

    \end{tikzpicture}
    \caption{Conjecture \ref{conj:anticoncentration} is that permanents of the form $\left|\Per (R+tW)\right|$ are $ \geq (n!/\poly(n))^{-1}$ for matrices $W$ whose entries are bounded by $1$ and for $t=O(\frac{1}{\sqrt n})$. This interpolates between $t=0,$ i.e.\ PACC \cite {Aaronson2013}, 
    and $t \geq \frac{1}{\text{polylog}(n)},$ where anticoncentration is proved unconditionally by \cite{ji2021approximating}, improving upon \cite{eldar-mehraban}. In other words, we conjecture that permanents along the thickened orange line are at least as anticoncentrated as $\left|\Per R\right|$ at $t=0.$}
    \label{fig:anticoncentration}
\end{figure}

\anticoncentration*

Intuitively, the statement is that permanents of nonzero-mean Gaussian matrices are at least as anticoncentrated as are zero-mean Gaussians. In fact, the only setting in which there exists a proof of anticoncentration for Gaussian permanents\footnote{Although there are proofs of so-called ``weak'' anticoncentration, these do not imply the stronger form of anticoncentration 
necessary for the reductions
made throughout the BosonSampling literature.}
is in the case of nonzero mean, in particular for $\mcal N(t,1)^{n\times n}$ matrices with $t$ at least $1/\poly\log n$ \cite{ji2021approximating, eldar-mehraban}. With this exception, all forms of anticoncentration for BosonSampling remain open to date, to the authors' knowledge. 

In order to prove Theorem \ref{thm:no-sampler}, we prove a robust worst-to-average-case reduction that synthesizes the techniques developed earlier: coefficient extraction, the square method, magnification, and the first rare events lemma, Lemma \ref{lem:probability-attenutation}. 

\begin{theorem}
\label{thm:hardness-of-probs}
It is \textsf{\#P}-hard to compute $\left|\Per R\right|^2$ for $R\sim\mcal N(0,1)^{n\times n}$ to $1/\poly(n)$ relative error, with probability at least $1-\exp(-O(n))$ over the choice of $R,$ assuming Conjecture \ref{conj:anticoncentration}.
    \end{theorem}

\begin{proof}[Proof of Thm.~\ref{thm:hardness-of-probs}]
Take an arbitrary constant $\eps>0$ and for ease of notation, define $k\coloneq \lfloor n^\eps \rfloor.$

Let $\mcal A $ be an algorithm that given as input $R\sim\mcal N(0,1)^{n\times n},$ approximates $\left|\Per R\right|^2$ to within $1/\poly(n)$ relative error, with probability at least $1-\exp(-O(n))$ over the choice of $R.$ Additionally, consider a ``worst-case'' block-diagonal matrix $W$ with an upper-left block $W'\in\{0,\pm 1\}^{k\times k}$, and a lower-right $(n-k)\times (n-k)$ block of the all $1$s matrix.

We will show that then there exists a $\mathsf{BPP}^\mathsf{NP^{\mcal A}}$ procedure that given as input \emph{any} such matrix $W,$ approximates $\left|\Per W'\right|$  to within small relative error, 
with success probability at least $\frac{2}{3}.$
The theorem statement follows immediately from the \textsf{\#P}-hardness of computing a multiplicative approximation to the permanent of a $\{0,\pm 1\}$ matrix.

Define the polynomial 
\begin{align}
    \left|\Per A(t)\right| \coloneq \left|\Per (A(0)+tW)\right|,
\end{align}
where $A(0)\sim \mcal N(0,1)^{n\times n}$ and $W$ is as above. Then $\left|\Per A(t) \right|$ is a degree $n$ polynomial in $t$ whose leading coefficient is  $\left|\Per W\right| = (n-k)! \left|\Per W'\right|.$
This polynomial is illustrated in Fig.\ \ref{fig:all-tricks}.

By Lemma \ref{lem:probability-attenutation}, if $\mcal A$ computes a $\pm\gamma$-approximation to $\left|\Per R\right|^2$ with probability at least $1-\beta$, then it computes a $\pm\gamma$-approximation to $\left|\Per A(t)\right|^2$ with probability at least $1-\sqrt{\beta\cdot e^{t^2 n^2}}.$ As in the theorem statement, we take $\beta=\exp(-O(n))$ so that $\mcal A$ has at least $1-1/\poly(n)$ probability to correctly compute $\left|\Per A(t^\ast)\right|^2$ where $t^\ast = O(1/\sqrt{n}).$ 
In particular, we will take $t^\ast=4e^{2.5}\cdot n^{k^{-1}}/\sqrt{n} =4e^{2.5} (1+o(1))/\sqrt{n}$, 
and $\beta = 
\exp(-16e^5 n - O(\log n)).$

Consider $O(n)$ equally spaced $\{t_i\}$ in the interval $[0,\Delta]$ for $\Delta=O(1/\sqrt n).$ By a union bound, all the points are correct to within $\pm\gamma$ with probability at least $1-1/\poly(n).$
Then the $\mathsf{BPP}^{\mathsf{NP}^{\mathcal A}}$ procedure is as follows: query the \textsf{NP} oracle for a degree $n$ polynomial $q$ such that $|q(t_i)|^2$ is $\pm \gamma$-close to the value obtained by $\mcal A$ for at least half of the points $\{t_i\}$. 
This admits a certificate that can be efficiently verified by checking each point $\{t_i\}$ for agreement between $\mcal A$ and $|q(t_i)|^2.$ Return as output 
$|q_n|.$ 

Finally, Lemma \ref{lem:square-trick-coex} guarantees that $|q_n|$ is a good multiplicative estimator of $\left|\Per W\right|.$ 
As $\mcal A$ obtains a $1/\poly(n)$ relative error approximation to $\left|\Per R\right|^2$, we have $\gamma=n!/\poly(n).$ Assuming Conjecture \ref{conj:anticoncentration}, $|\Per A(t_i)|\geq \sqrt{n!}/\poly(n)$ on the set of points at which the \textsf{NP} oracle and $\mcal A$ agree. By construction, $\left|\Per W\right| = (n-k)! \left|\Per W'\right|
= n!\exp(-k\log n+O(k\log k)).$ Recalling from above that $t^\ast=4e^{2.5}\cdot n^{k^{-1}}/\sqrt{n} =4e^{2.5} (1+o(1))/\sqrt{n}$ and substituting all these values into Eq.\ \ref{eq:sq-trick-coex} of Lemma \ref{lem:square-trick-coex}, we find at last that $||p_n| - \left|\Per W\right|| \leq  \frac{1}{\poly(n)}\vert\Per W\vert.$
Overall, we have a $\mathsf{BPP}^{\mathsf{NP^{\mcal A}}}$ procedure to multiplicatively estimate $\left|\Per W\right|,$ which concludes the proof.
\end{proof}

We have now developed the machinery to prove the main result of this section, Theorem \ref{thm:no-sampler}. Our robust worst-to-average-case reduction in Theorem \ref{thm:hardness-of-probs}, combined with the second rare events lemma, Lemma \ref{prp:small-small}, allows us to prove the first nontrivial hardness of \emph{sampling} result for average-case BosonSampling. 
\nosampler*

\begin{proof}[Proof of Theorem \ref{thm:no-sampler}]
Suppose such a sampler exists. Then, given as input a Haar-random orthogonal matrix, to within $1/\poly(n)$ relative error one can compute the squared permanent of the submatrix corresponding to a given output probability in $\mathsf{BPP}^\mathsf{NP}$ via Stockmeyer's approximate counting algorithm \cite{stockmeyer1983complexity}. Next we invoke Proposition \ref{prp:small-small} (proved in Appendix \ref{sec:rare-events}) which says that ``rare event'' $E$ sampled as the $n\times n$ submatrix of an $m\times m$ Haar-random orthogonal matrix that occurs with probability $\Prob_S(E)\leq\delta$,  occurs with probability $\Prob_G(E)\leq \delta\cdot \exp(O( \sqrt{n})) + O(\exp(-n)) $ if $E$ is instead sampled from the i.i.d.\ Gaussian measure. Consequently $\Prob_S(E)\leq\exp(-O(n))$ implies \begin{align*}
\Prob_G(E)&\leq \exp(-O(n))\cdot \exp(O( \sqrt{n})) + O(\exp(-n)) \\&\leq \exp(-O(n)).    
\end{align*} This suffices to show that the sampler of the theorem statement likewise has $1-\exp(-O(n))$ success probability to correctly compute the squared permanents a matrix drawn from the $n\times n$ i.i.d.\ Gaussian matrices to within $1/\poly(n)$ relative error. By Theorem \ref{thm:hardness-of-probs}, doing so is \textsf{\#P}-hard. Finally by Toda's theorem, this collapses \textsf{PH}.

\end{proof}

Finally, we show that the failure probability in Theorems \ref{thm:no-sampler} and \ref{thm:hardness-of-probs} exponentially improves upon the ``trivial'' algorithm that computes $\Per W$ directly. 
The intuition is that an algorithm to compute $\Per((1-t)R+t W)$ for $t$ very close to $1$ would need failure probability at most $\exp(-\tilde{O}(n^3))$, as it takes $\tilde{O}(n^3)$ bits to specify $W,$ an $n\times n$ matrix of reals specified to $\tilde{O}(n)$ bits of precision. By comparison, Theorem \ref{thm:hardness-of-probs} pertains to a sampler that fails with probability at most $\exp(-O(n))$. 

We formalize this intuition as follows: Lemma \ref{lem:nontriviality} shows that if matrices $A$ and $A+\delta B$ with entries bounded by 1 are sufficiently close, i.e.\ for sufficiently small $\delta,$ then their permanents are also close. As we show $\delta$ to be $\exp(-\tilde{O}(n^3))$, an algorithm that computes the permanents of all but this tiny fraction of matrices correctly is guaranteed to approximate the permanents of the remaining matrices, as well.

\begin{lemma}\label{lem:nontriviality}
Suppose that $A,B$ are matrices with entries $|a_{ij}|, |b_{ij}|\leq 1$ and $\delta < \frac{1}{100} n^{-1} (n!)^{-1}$.
Then 
\[
|\Per(A) - \Per(A+\delta B)| \leq 1.
\]
\end{lemma}
\begin{proof}
\begin{align*}
|\Per(A+\delta B)-\Per(A)|
&=|\sum_\pi [\prod_{i=1}^n
a_{i\pi(i)}
- \prod_{i=1}^n (a_{i\pi(i)}+\delta b_{i\pi(i)})]| \\
&\leq 
\sum_\pi |\prod_{i=1}^n (1+\delta b_{i\pi(i)}) - 1| \\
\end{align*}
Each term in the above sum is bounded by 
$(1+\delta)^N-1$.  Thus, as there are $n!$ terms we compute
\begin{align*}
|\Per(A+\delta B)-\Per(A)|
&\leq n! ((1+\delta)^N-1) \\ 
&\leq n! (|\exp(N\delta)-1|
+ |\exp(N\delta) - (1+\delta)^N|) \\
&< 1.
\end{align*}
\end{proof}
The key point is that if $g$  is an approximation to the permanent that is wrong on $\frac12$ of the cube of width $n^{-1}(n!)^{-1})$ centered at some matrix $A_0$, then  it is in particular wrong on a set of volume $(n^{-1}(n!)^{-1})^{n^2}$, which is to say $\exp^{-Cn^3\log n}$.  That means there is a trivial answer to the question only for error probabilities like $\exp(-\tilde O(n^3))$ rather than $\exp(-O(n))$.

\section*{Acknowledgements}\label{sec:ack}

We thank Scott Aaronson, Daniel Grier, Hari Krovi, and Umesh Vazirani for insightful discussions.
A.B.~ and I.D.~were supported in part by the AFOSR under grants FA9550-21-1-0392 and FA9550-24-1-0089.
A.B.~was supported in part by the DOE QuantISED grant DE-SC0020360 and by the U.S. DOE Office of Science under Award Number DE-SC0020377.
I.D.~was supported in part by the Lieberman Fellowship.
B.F.~acknowledges support from the National Science Foundation under Grant CCF-2044923 (CAREER), by the U.S. Department of Energy, Office of Science, National Quantum Information Science Research Centers (Q-NEXT) and by the DOE QuantISED grant DE-SC0020360.  
F.H.~was supported by the John and Fannie Hertz fellowship and NSF award DMS-2303094.
This work was done in part while the authors were visiting the Simons Institute for the Theory of Computing, supported by DOE QSA grant \#FP00010905 and NSF QLCI Grant No. 2016245.

\newpage

\section*{Appendices}

\appendix

\section{How far can you shift and scale i.i.d.\ Gaussian matrices?}\label{sec:pinsker}

In this section, we quantify how much an i.i.d.\ Gaussian matrix $R$ is perturbed under ``shift'' and ``scale,'' namely dilation by $(1-t)$ and translation by $tW$ for $t\in[0,1]$ and worst-case matrix $W.$ In other words, for what values of $t$ is the distribution over $(1-t)R+tW$ a constant total variation distance from that of unperturbed distribution over $R$? 

The proof proceeds by an explicit calculation of the KL divergence between two Gaussians followed by Pinsker's inequality. This is observed in \cite{krovi2022}, with similar calculations appearing in \cite{jiang-ma} and \cite{Chabaud2022quantuminspired}. For completeness, we give the proof here. 

It follows immediately from the proof that distributions that are only shifted, not scaled, likewise give $O(nt)$ total variation distance---this is the case for coefficient extraction. For Gaussians under shifts only, \cite{Aaronson2013} (Lemma 48) also calculates a total variation distance of $O(nt)$ but by a different method.

\begin{lemma}[Autocorrelation of Gaussian distribution] \label{lem:tvd-pinsker}
     $\|\mathcal{D}_{(1-t)R+tW}-\mathcal D_{R}\|_\text{TVD}\leq O(nt).$
\end{lemma}

\begin{proof}
We obtain an upper bound on total variation distance via Pinsker's inequality: 
    \begin{align}\label{eq:pinsker}
    \sqrt{2} \|\mathcal{D}_{(1-t)R+tW}-\mathcal D_{R}\|_\text{TVD}\leq \sqrt{D_{\text{KL}}(\mathcal{D}_{(1-t)R+tW},\mathcal D_{R})},
\end{align}
where on the right we have the KL divergence. By definition $R_{ij}\sim\mcal N(0,1)$, so $((1-t)R+tW)_{ij}\sim \mathcal N(t w_{ij},(1-t)^2)$. 
The KL divergence between two Gaussians is
\begin{align}
    D_{\text{KL}}(\mathcal{N}(\mu_0,\sigma_0),~\mathcal{N}(\mu_1,\sigma_1)) = \frac{(\mu_0-\mu_1)^2+\sigma_0^2}{2\sigma_1^2}+\log\frac{\sigma_1}{\sigma_0}-\frac{1}{2}.
\end{align}
So \begin{align}
    D_{\text{KL}}(\mcal N(t w_{ij},(1-t)^2),~\mathcal{N}(0,1))&=
    O(t^2),
\end{align}
as $w_{ij}=O(1).$ Note that KL divergence is not symmetric so the order above matters. 

Recalling that the KL divergence is additive for independent distributions, the RHS of Eq.~\ref{eq:pinsker} is 
\begin{align}
    \sqrt{D_{\text{KL}}(\mathcal{D}_{(1-t)R+tW},\mathcal D_{R})} = O(nt).
\end{align}

\end{proof}

Notably, in sharp contrast to the simple bound above, the analogous bounds for the shift-and-scale behavior of distributions in the ``low-mode'' or saturated limit regime of BosonSampling are highly nontrivial and are detailed in \cite{low-modes}.

\section{Corollaries for Random Circuit Sampling}\label{sec:rcs}

In this section we describe corollaries of our techniques for Random Circuit Sampling. First, we prove the discrete Remez inequality and give a much simpler proof of Robust Berlekamp-Welch, introduced in \cite{Bouland2021}.

\input{robust-bw.tex}

\subsection*{Square method for extrapolation}\label{sec:square-extrapolation}

\begin{lemma}[The square method for extrapolation]
\label{lem:square-trick-extrapolation}
Let $p$ and $q$ be real-valued polynomials of degree $d$, and let
$S\subset [0,1]$ be a $\delta$-separated set of points with $|S|=2d+1$.  Then,
\begin{equation}
\label{eq:square-trick-extrap-bound}
|p^2(1) - q^2(1)| \leq E |p(1)| + E^2
\end{equation}
where
\[
E = (e^2 (d\delta)^{-1})^d \max_{x\in S} \frac{|p^2(x)-q^2(x)|}{|p(x)|}.
\]
\end{lemma}

\begin{proof}[Proof of Lemma~\ref{lem:square-trick-extrapolation}]

By Lemma~\ref{lem:difference-of-squares}, we can conclude that
\[
||p(x)|-|q(x)|| \leq |p(x)|^{-1} |p(x)^2 - q(x)^2|
\]
for each data point $j$.  In particular, for each $x\in S$ there exists a sign
$\sigma_x\in\{\pm 1\}$ such that
\[
|p(x)-\sigma_xq(x)| \leq |p(x)|^{-1} |p(x)^2 - q(x)^2|.
\]
Let $S^+ = \{x\mid \sigma_x=1\}$ and $S^-=\{x\mid\sigma_x = -1\}$.  Since
$|S^+\cup S^-|=2d+1$, it follows that either $|S^+|\geq d+1$ or $|S^-|\geq d+1$.  Without
loss of generality suppose that $|S^+|\geq d+1$.  Then $S^+$ is
also a $\delta$-separated set of points, so by Lemma~\ref{lem:discrete-remez},
\[
|p(1) - q(1)|
\leq
(e^2 d\delta)^{-d} \max_{x\in S^+} \frac{|p^2(x) - q^2(x)|}{|p(x)|}
\leq E
\]
To obtain~\eqref{eq:square-trick-extrap-bound} we use the triangle inequality
to bound $|q(1)|\leq |p(1)|+|p(1)-q(1)|$ and write
\[
|p^2(1)-q^2(1)|
= |p(1)-q(1)||p(1)+q(1)|
\leq |p(1)-q(1)|(|p(1)| + |p(1)-q(1)|).
\]
\end{proof}

\subsection*{Dilution for Random Circuit Sampling: Corollary \ref{cor:rcs-dilution}}\label{sec:dil-rcs}

In this section, we show that the argument for Thm.\ \ref{thm:dilution-intro} can be adapted to Random Circuit Sampling, as well. 
The proof follows readily from the dilution argument illustrated in Fig.\ \ref{fig:rcs-dilution} combined with well-established machinery from \cite{movassagh_quantum_2020, Kondo2021_robustness, Bouland2021, movassagh2023hardness}.

\rcsdilution*

\begin{proof}[Proof sketch]

As in the proof of Theorem \ref{thm:dilution-intro}, we note that output probabilities of a circuit have a polynomial structure, in this case coming from the Feynman path integral. Depicted in Fig.\ \ref{fig:dilution}, we take a random circuit supported on $n$ qubits, calling the circuit supported on the first $n^\varepsilon$ qubits $R_A$ and on the latter $n-n^\varepsilon$ qubits $R_B.$ 

Then, we perturb only the gates supported on the first $n^\varepsilon$ qubits, circuit $R_A,$ to a worst-case circuit $W_A$ by the Cayley transform parametrized by $\theta$ \cite{movassagh_quantum_2020,movassagh2023hardness}. Notably, two conditions hold: output probabilities of the $\theta$-perturbed random circuit family are a low-degree rational function in $\theta,$ in particular with degree $(O(n^\varepsilon),O(n^\varepsilon)),$ and moreover the total variation distance between the initial and $\theta$-perturbed distributions is $O(k\theta).$ 

By Theorem \ref{thm:rbw}, an algorithm $\mcal A$ to compute output probabilities from the unperturbed circuit up to additive error $\gamma$ can be in $\mathsf{BPP}^{\mathsf{NP}^{\mcal A}}$ converted into Robust Berlekamp-Welch extrapolation that computes output probabilities of \emph{any} circuit, e.g.,\ a Fourier Sampling circuit, up to additive error $\gamma\cdot 2^{n+O(n^\delta)}.$ This is $\textsf{\#P}$-hard, completing the proof.
\end{proof}

\begin{remark}[\textit{Depth barrier}]
Referring to Lemma \ref{lem:square-trick-extrapolation}, our extrapolation bounds make use of lower bounds on the polynomial close to $\theta=0.$ In this way we invoke anticoncentration, proved for random circuits at log depth \cite{logANTICONCENTRATION}. In doing so, Corollary \ref{cor:rcs-dilution} is depth-sensitive, i.e.\ requires sufficiently deep circuits, and thus overcomes the depth barrier described in the Introduction.
\end{remark}

\begin{remark}[\textit{Born-rule barrier}]
    Recall the ``Born-rule'' barrier identified by Krovi \cite{krovi2022}, namely that the additive error needed to prove the hardness of average-case sampling ($~2^{-n}$) is already larger than the additive error known to be hard in the worst case ($2^{-2n}$, which is derived from the Born rule by squaring the output amplitude of a Quantum Fourier Sampling circuit). Without dilution, it is seemingly impossible to prove a worst-to-average case reduction in which the additive error in the average case \emph{is larger} than the additive error we need to obtain in the worst case. Corollary \ref{cor:rcs-dilution} overcomes the Born-rule barrier by simply 
    ``diluting'' the worst-case instance to be polynomially smaller than the average-case instance.
\end{remark}

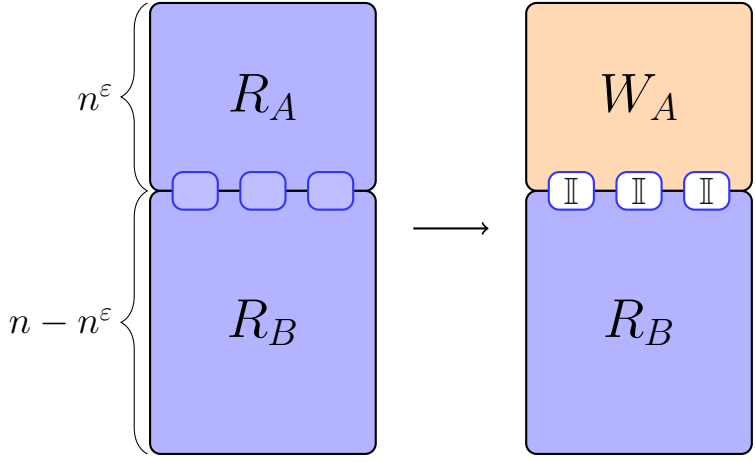
\begin{figure}[ht]
\hspace{2cm} 
\begin{tikzpicture} 

\draw[thick, rounded corners, fill=blue!30] (0,3.5) rectangle (3,6); 
\draw[thick, rounded corners, fill=blue!30] (0,0) rectangle (3,3.5);

\node at (1.5,4.75) {\huge $R_A$};
\node at (1.5,1.75) {\huge $R_B$};

\draw[rounded corners, fill=blue!25, draw=blue!80, thick] (0.3,3.25) rectangle (0.9,3.75);

\draw[rounded corners, fill=blue!25, draw=blue!80, thick] (1.2,3.25) rectangle (1.8,3.75);

\draw[rounded corners, fill=blue!25, draw=blue!80, thick] (2.1,3.25) rectangle (2.7,3.75);

\draw[decorate,decoration={brace,amplitude=10pt,mirror},xshift=-1pt] (0,6) -- (0,3.5) node[midway,left,xshift=-9pt] {\Large $n^\varepsilon$};
\draw[decorate,decoration={brace,amplitude=10pt,mirror},xshift=-1pt] (0,3.5) -- (0,0) node[midway,left,xshift=-9pt] {\Large $n - n^\varepsilon$};

\draw[thick, ->] (3.5,3) -- (4.5,3);

\draw[thick, rounded corners, fill=orange!30] (5,3.5) rectangle (8,6); 
\draw[thick, rounded corners, fill=blue!30] (5,0) rectangle (8,3.5);

\node at (6.5,4.75) {\huge $W_A$};
\node at (6.5,1.75) {\huge $R_B$};

\draw[rounded corners, fill=white, 
draw=blue!80, thick] (5.3,3.25) rectangle (5.9,3.75);
\node[text=gray!20!black] at (5.6,3.5) {\Large $\mathbb{I}$};
\draw[rounded corners, fill=white, draw=blue!80, thick] (6.2,3.25) rectangle (6.8,3.75);
\node[text=gray!20!black] at (6.5,3.5) {\Large $\mathbb{I}$};
\draw[rounded corners, fill=white, draw=blue!80, thick] (7.1,3.25) rectangle (7.7,3.75);
\node[text=gray!20!black] at (7.4,3.5) {\Large $\mathbb{I}$};

\end{tikzpicture}

    \caption{In Corollary \ref{cor:rcs-dilution}, we take a random circuit supported on $n$ qubits and perturb only the circuit supported on the first $n^\varepsilon$ qubits, $R_A,$ to a worst-case circuit $W_A$ by the Cayley transform. The ``bridge'' gates on the interface between $R_A$ and $R_B$ transform into Identity gates so that the output probability on the righthand circuit factorizes.}
    \label{fig:rcs-dilution}
\end{figure}

\input{orth-vs-gauss}

\section{Numerical evidence for Conjecture \ref{conj:anticoncentration}}\label{sec:num}
In this section, we provide brief numerical evidence for Conjecture \ref{conj:anticoncentration}:

\anticoncentration*

In particular, we compute permanents of several ensembles of Gaussian matrices with varying means, and plot their distribution. We observe similar distributions on the minimum non-zero permanent observed for all means tested. 

\begin{figure}[h]
    \centering
    \includegraphics[width=0.75\linewidth]{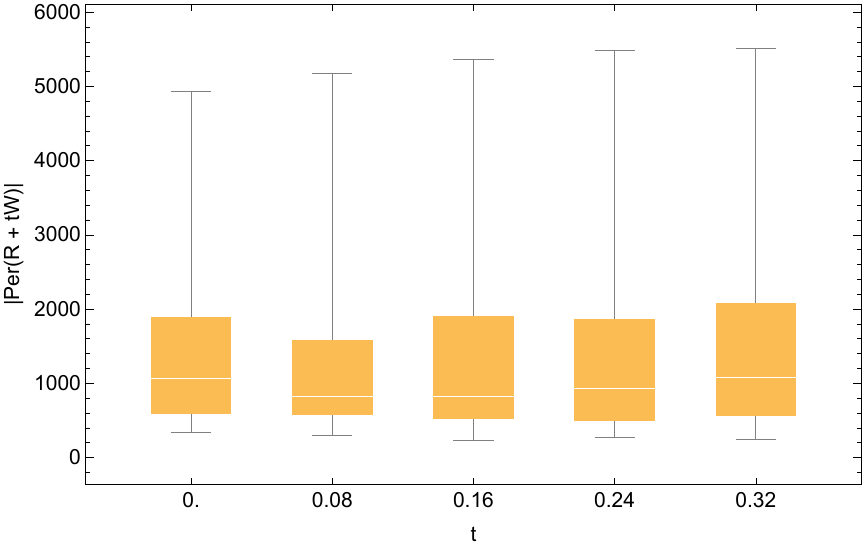}
    \caption{Box plots for the distribution of $\left|\Per(R+tW)\right|$ of the form depicted in Fig.\ \ref{fig:all-tricks} for $n=10$ and $k=n^\varepsilon=5.$ For five equally spaced values of $t\in[0,1/\sqrt{n}],$     
    we randomly generate $30$ such $R$ and $W$. Notably, the box plots show remarkably little variation for increasing $t$ in the relevant range, and in particular the lower bound for $t=0$ holds for shifted $t,$ as conjectured.}
    \label{fig:anticoncentration-statistics}
\end{figure}

\section{Square method for complex polynomials} 
\label{sec:complex-squares}
\input{complex-squares.tex}

\newpage 
\input{main.bbl}

\end{document}

%% file: top-coefficient.tex
\subsection{Coefficient extraction}\label{sec:co-ex}

The overall scheme of past worst-to-average-case reductions for BosonSampling is an interpolation argument inspired by Lipton's self-reducibility of the permanent, which exploits its polynomial structure to show that average-case instances are as hard as in the worst case \cite{lipton1991}. In particular, by taking a convex combination in variable $t$ of an average-case instance and a worst-case instance, the permanent is a univariate polynomial in $t.$ Then, by estimating values of the polynomial for small $t$ by the average-case algorithm, one can extrapolate to $t=1,$ the permanents of which are \textsf{\#P}-hard.

In this way, prior work has 
used the polynomial $\Per((1-t)R + tW)$ where $R$ is a Gaussian random matrix and $W$ is a
worst-case matrix. 
On the other hand, the polynomial $\Per(R+tW)$ also records
information about $\Per(W)$ as the highest order term is $t^n \Per(W)$. 
We use this observation to provide an alternative way to perform a worst-to-average-case reduction for computation of the permanent. We can sample the values of this polynomial
up to $t=O(n^{-1})$ because translation does not change the probability distribution quickly (see Appendix \ref{sec:pinsker}).

Our main new technical ingredient is a way to recover the top coefficient of a polynomial from its values on an interval.

\begin{lemma}
\label{lem:top-coeff}
Let $p(x) = \sum_{j=0}^d p_j x^j$ be a polynomial of degree $d$ satisfying
\[
\sup_{x\in[-\ell,\ell]} |p(x)| \leq \alpha.
\]
Then,
\[
|p_d| \leq 2^{d+1} \ell^{-d} \alpha.
\]
\end{lemma}

One should think of $p(x)$ as the difference between the true permanent polynomial and the approximate polynomial provided by the \textsf{NP} oracle in the reduction. The following lemma then provides a bound on the error in the approximate polynomial's top coefficient, which encodes the worst-case permanent, i.e.\ $\Per(W)$ above. This induces an additive error bound on the worst-case permanent. 

\begin{proof}
By rescaling the inputs, it suffices to prove the result with $\ell=1$.

Let $T_n(x)$ be the $n$-th Chebyshev polynomial.  These polynomials satisfy the orthogonality
relation
\[
\int_{-1}^1 T_n(x) T_m(x) \frac{dx}{\sqrt{1-x^2}} =
\begin{cases}
0, &n\not=m \\
\pi, &n=m=0 \\
\frac{\pi}{2}, &n=m\not=0.
\end{cases}
\]
Since $\operatorname*{span}\{1,x,\cdots,x^n\}=\operatorname*{span}\{T_0,T_1,\cdots,T_n\}$,
it also follows that
\[
\int_{-1}^1 q(x) T_d(x) \frac{dx}{\sqrt{1-x^2}} = 0
\]
whenever $q$ is a polynomial of degree at most $d-1$.  Since the coefficient of $x^n$ in $T_n$ is $2^n$,
the polynomial $p - 2^{-d} p_d T_d$ is a polynomial of degree $d-1$ so that
\begin{align*}
\int_{-1}^1 (p(x) - 2^{-d}p_d T_d(x)) T_d(x) \frac{dx}{\sqrt{1-x^2}}  = 0.
\end{align*}
Rearranging and using the orthogonality relations above, this becomes
\[
2^{-d-1}\pi p_d = \int_{-1}^1 p(x) T_d(x) \frac{dx}{\sqrt{1-x^2}}.
\]
On the other hand, using the uniform bound on $p(x)$ we can bound the latter integral as follows:
\[
\big|\int_{-1}^1 p(x) T_d(x) \frac{dx}{\sqrt{1-x^2}}\big|
\leq \alpha \int_{-1}^1 \frac{dx}{\sqrt{1-x^2}} = \pi \alpha.
\]
\end{proof}

%% file: real-squares.tex
\subsection{Square method}\label{sec:square-trick}

Observe that the polynomial
$|\Per(R_t)|^2\coloneq\left|\Per(R+tW)\right|^2$ is not merely a polynomial of degree $2n$,  
but rather
the square of a polynomial of degree $n$.  It is natural to then ask whether this observation
can be used to reduce the effective degree of the extrapolation to $n$ instead of $2n$.

We suppose that we have some approximate values of a square polynomial $p^2$, and we use
an $\NP$ oracle to find some square $q^2$ that agrees with the approximate values.
Then one expects that either $p\approx +q$ or $p\approx -q$ on these values.  Notably in the case of complex polynomials, rather than a sign ambiguity there is a phase ambiguity. The consequences of this are explored in Appendix \ref{sec:complex-squares}. Throughout the main body of the text, we restrict our attention to real polynomials.

Before we prove Lemma~\ref{lem:square-trick-coex} we record an elementary fact.
\begin{lemma}
\label{lem:difference-of-squares}
Let $p,q\in\Real$ be real numbers satisfying
\[
|p^2 - q^2| < \delta.
\]
Then $||p|-|q||< |p|^{-1}\delta$.
\end{lemma}
\begin{proof}
We can assume without loss of generality that $p$ and $q$ are positive.  Then
$|p+q|>|p|$, so
\[
|p-q| \leq |p|^{-1}|p-q||p+q| = |p|^{-1}|p^2-q^2| < |p|^{-1}\delta.
\]
\end{proof}

Here we introduce the square method in the setting of coefficient extraction. To do so, we need the following discrete Remez inequality. It is proved in Appendix \ref{sec:rbw}, as is the square method for extrapolation.
\begin{lemma}[Discrete Remez inequality]
\label{lem:discrete-remez-main-body}
Let $\{x_j\}_{j=0}^{d}\subset [-\ell,\ell]$ be a $\delta$-separated set of points, meaning that
$|x_i-x_j| \geq \delta$ for $i\not=j$.  Then if $p$ is a degree-$d$ polynomial
\[
\sup_{[-\ell,\ell]}|p(x)| \leq (2e^2(\delta d)^{-1}\ell)^{d} \max_{0\leq j\leq d} |p(x_j)|.
\]
\end{lemma}

\begin{lemma}[The square method for coefficient extraction]
\label{lem:square-trick-coex}
Let $p$ and $q$ be real-valued polynomials of degree $d$ and
let $S$ be a $\delta$-separated collection of points in $[-\ell,\ell]$ with
$|S|\geq 2d+1$.
Suppose moreover that
\[
\sup_{x\in S} |p^2(x) - q^2(x)| \leq \gamma
\]
and $\inf_{x\in S} |q(x)| \geq K$.  Then
\begin{equation} \label{eq:sq-trick-coex}
||p_d| - |q_d|| \leq 2^{2d+1} e^{2d} (d\delta)^{-d} K^{-1}\gamma.
\end{equation}
\end{lemma}

\begin{proof}[Proof of Lemma~\ref{lem:square-trick-coex}]
By Lemma~\ref{lem:difference-of-squares}, we can conclude that
\[
||p(x)|-|q(x)|| \leq K^{-1} |p(x)^2 - q(x)^2|
\]
for all $x\in S$.  In particular, for each $x\in S$ there exists a sign
$\sigma_x\in\{\pm 1\}$ such that
\[
|p(x)-\sigma_xq(x)| \leq K^{-1} |p(x)^2 - q(x)^2|.
\]
Let $S^+ = \{x\in[-\ell,\ell]\mid \sigma_x=1\}$ and $S^-=\{x\in[-\ell,\ell]\mid\sigma_x = -1\}$.
At least one of these sets must contain more than $d+1$ points,
so without loss of generality suppose that $|S^+|\geq d+1$.
Then $S^+$ is
also a $\delta$-separated set of points, so by Lemma~\ref{lem:discrete-remez-main-body},
\[
\sup_{x\in[-\ell,\ell]} |p(x)-q(x)|
\leq 2^d e^{2d} (d\delta)^{-d} \ell^d K^{-1}\gamma
\]
By Lemma~\ref{lem:top-coeff} we have
\[
|p_d - q_d| \leq 2^{2d+1} e^{2d} (d\delta)^{-d}  K^{-1}\gamma.
\]
\end{proof}

We develop the square method for Random Circuit Sampling in Appendix \ref{sec:square-extrapolation}, where we use extrapolation rather than coefficient extraction. There, we can extrapolate the values of $p$ itself rather
than $p^2,$ and use the discrete Remez inequality 
to bound the extrapolation blowup induced on $p^2$.  

%% file: robust-bw.tex
\subsection*{Discrete Remez Inequality and Robust Berlekamp-Welch}\label{sec:rbw}

\noindent Powering both extrapolation and coefficient extraction is the discrete Remez inequality, proved in this section.

\discreteremez*

\begin{proof}
Using Lagrange interpolation, we can write
\[
p(x) = \sum_{j=0}^d p(x_j) \frac{\prod_{k\not=j}(x-x_k)}{\prod_{k\not=j} (x_j-x_k)}.
\]
To see that this identity holds, observe that it holds at any $x_j$ and that both sides are
polynomials of degree $d$.  Substituting $x=L$ and observing $|L-x_k|\leq L$, we obtain the bound
\[
|p(L)| \leq  L^d \max_{0\leq j\leq d}|p(x_j)|  \max_{j} \prod_{k\not=j} |x_j-x_k|^{-1}.
\]
It remains to show that
\begin{equation}
\label{eq:denominator-bd}
\max_{j} \prod_{k\not=j} |x_j-x_k|^{-1} \leq e^{2d} (\delta d)^{-d},
\end{equation}
which by taking logarithms is equivalent to
\[
\max_j \sum_{k\not= j} \log |x_j-x_k|^{-1} \leq 2d +  d\log (d\delta)^{-1}.
\]
We use the layer-cake formula to estimate the sum, writing
\begin{align*}
\sum_{k\not=j} \log |x_j-x_k|^{-1}
&= \sum_{k\not= j} \int_0^{\log |x_j-x_k|^{-1}} \diff t \\
&= \int_0^\infty \#\{k \mid \log |x_j-x_k|^{-1} > t \}\diff t \\
&= \int_0^1 s^{-1}\#\{k \mid |x_j-x_k| < s \}\diff s.
\end{align*}
The second step follows from Fubini's theorem, and the last step from the change of variables
$s=e^{-t}$.
The $\delta$-separated hypothesis on $x_j$ implies
\[
\#\{k \mid |x_j-x_k| < s \}
\leq \begin{cases}
0, & s<\delta \\
2\delta^{-1}s, & \delta \leq s \leq d\delta\\
d, & s > d\delta.
\end{cases}
\]
Therefore
\[
\sum_{k\not=j} \log |x_j-x_k|^{-1}
\leq \int_\delta^{d\delta} 2\delta^{-1} \diff s
+ \int_{d\delta}^1 ds^{-1}\diff s
= 2d + d\log(d\delta)^{-1},
\]
which concludes our proof of~\eqref{eq:denominator-bd}.
\end{proof}

As a consequence of Lemma~\ref{lem:discrete-remez}, we obtain a simpler proof of
Robust Berlekamp-Welch, which was initially developed in \cite{Bouland2021}. 
\begin{theorem}[Robust Berlekamp-Welch bound]
\label{thm:rbw}
Let $D=\{(x_i,y_i)\}_{i=1}^{M}$ be a set of $2(d+1) < M < 100d$ data points with $x_i$ evenly spaced
on the interval $[0,\Delta]$.  Suppose that $P_1$ and $P_2$ are degree-$d$ polynomials
which satisfy
\begin{equation}
   \#\{j \mid |P_a(x_j) - y_j| \geq \delta\} < M/4 \label{eq:small-fraction-wrong} 
\end{equation}

for $a=1,2$.  Then
\begin{equation}
|P_1(1) - P_2(1)| \leq (C \Delta^{-1})^d \delta. \label{eq:rbw-bound}
\end{equation}
\end{theorem}
\begin{proof}
The set on which $P_1$ and $P_2$ agree has at least $M/2 > (d+1)$ points.  These points
are $O(\Delta d^{-1})$-separated.  The conclusion follows from an application of
Lemma~\ref{lem:discrete-remez}.
\end{proof}

The key observation is that Theorem \ref{thm:rbw} can be turned into an algorithm in $\mathsf{BPP}^{\mathsf{NP}}$ that carries out extrapolation.
In the reduction, $P_2$ is supplied by an \textsf{NP} oracle, where Eq.\ \ref{eq:small-fraction-wrong} is the efficiently-verifiable predicate.
Thus the algorithmic interpretation of Theorem \ref{thm:rbw} is that in $\mathsf{P}^{\mathsf{NP}}$, one can estimate at $t=1$ a polynomial given faraway points close to $t=0,$ even when a constant fraction of the points are utterly corrupted. This is admissible because our use of Stockmeyer approximate counting already necessitates a $\mathsf{BPP}^\mathsf{NP}$ reduction. The point is that 
because the worst-case is $\textsf{\#P}$-hard, a reduction at any finite level of \textsf{PH} induces its collapse.

%% file: orth-vs-gauss.tex
\section{Rare events for orthogonal submatrices and i.i.d.\ Gaussians}\label{sec:rare-events}
We consider two models of $n\times n$ random matrices.  The first is a Gaussian matrix $X^n$ with
independent (real) entries of variance $n^{-1}$.  It has a probability density given
by
\[
p_G(X) = Z_G^{-1}(n) ( \prod_{i\in[n]} \exp(-n\lambda_i(X^TX)/2)
\]
where $\lambda_i(A)$ is the $i$-th eigenvalue of $A$.  The factor of $n$ comes from the normalization we apply, and $Z_G^{-1}(n)$ is
a normalization constant so 
that
\[
\int_{\Real^{n\times n}} p_G(X)\diff X = 1.
\]

The second model is that of a $n\times n$ submatrix of a Haar-random $m\times m$ orthogonal matrix.
We rescale by $\sqrt{m/n}$ so that the individual entries have variance $n^{-1}$.
Then for $m\geq 2n$ the probability density takes the form (see \cite{jiang-ma}, Lemma 2.1)
\[
p_S(X) = Z_S^{-1}(n,m) \prod_{i\in[n]} (1-n\lambda_i(X^TX)/m)^{(m-2n)/2} \One_{\lambda_i \leq m/n}.
\]

Our main result in this section is the following:
\smallsmall*

The proposition will follow from three main facts.  The first relates the normalization constants
$c_G$ to $c_S$:
\begin{lemma}
\label{lem:Z-compare}
There exists a constant $C$ such that for $n>C$ and $m>Cn^2$,
the normalization constants $Z_G(n)$ and $Z_S(n,m)$ satisfy
\[
\frac{1}{10} \leq Z_G(n)/Z_S(n,m) \leq 10.
\]
\end{lemma}

The second fact relates the densities $p_G$ and $p_S$ directly, up to the normalization $c_G/c_S$:
\begin{lemma}
\label{lem:pGS-compare}
Suppose that $\lambda_{max}(X^TX) \leq K \leq \frac{m}{10n}$.  Then
\[
|\log(\frac{p_G(X)}{p_S(X)}) - \log(Z_S(n,m)/Z_G(n)| \leq
K^2 \frac{n^3}{m^2} +  K^3 \frac{n^4}{m^3} +
\frac{n^2}{m} |\trace[(X^TX)^2 - 2X^TX]|.
\]
\end{lemma}

The third fact we need is that the right hand side above is often small for $X$ sampled from
the independent Gaussian distribution.
\begin{lemma}
\label{lem:gauss-conc}
let $X$ be a Gaussian $n\times n$ matrix with independent entries of variance $n^{-1}$.  Then:
\begin{equation}
\label{eq:max-conc}
\Prob( \lambda_{max}(X^TX) > 3+t ) \leq \exp(-nt^2/2)
\end{equation}
Moreover,
\begin{equation}
\label{eq:trac-conc}
\Prob( |\trace[(X^TX)^2 - 2 X^TX]| > 100\sqrt{t}) \leq \exp(-t) + \exp(-n)
\end{equation}
\end{lemma}

Before we proceed to the proofs of these lemmas we show how to combine them to deduce
Proposition~\ref{prp:small-small}.
\begin{proof}[Proof of Proposition~\ref{prp:small-small} using Lemmas~\ref{lem:Z-compare}-\ref{lem:gauss-conc}]
We write
\[
E \subset (E \cap \Omega_{\rm good}) \cup \Omega_{\rm bad},
\]
where we set $\Omega_{\rm bad} = \Omega_{\rm good}^c$ and $\Omega_{\rm good}$ is the set of matrices
satisfying
\[
\Omega_{\rm good} := \{\lambda_{max}(X) \leq 4 \} \cap \{\trace[ (X^TX)^2 - 2X^TX] \leq 100 n^{\alpha/2}\}.
\]
Then by Lemma~\ref{lem:gauss-conc} and the fact that $\alpha \leq 1$ we have
\[
\Prob_G(\Omega_{\rm bad}) \leq 2\exp(-n) + \exp(-n^\alpha) \leq 3\exp(-n^\alpha),
\]
so using a union bound we have
\[
\Prob_G(E) \leq \Prob_G(E\cap \Omega_{\rm good}) + 3\exp(-n^\alpha).
\]

Next we estimate $\Prob_G(E\cap \Omega_{\rm good})$ by reweighting the probability measure:
\begin{align*}
\Prob_G(E\cap \Omega_{\rm good})
&= \int_{E\cap \Omega_{\rm good}} p_G(X)\diff X \\
&= \int_{E\cap \Omega_{\rm good}} \frac{p_G(X)}{p_S(X)} p_S(X) \diff X.
\end{align*}
For $X\in \Omega_{\rm good}$ we have by Lemma~\ref{lem:pGS-compare} the inequality
\begin{align*}
\frac{p_G(X)}{p_S(X)} &\leq \frac{c_G}{c_S} \exp(16\frac{n^3}{m^2} + 64\frac{n^4}{m^3}
+ 100 \frac{n^2}{m}  n^{\alpha/2}).
\end{align*}
By Lemma~\ref{lem:Z-compare} we have $\frac{c_G}{c_S} \leq 10$, and then using that $m>Cn^2$ and
$n>C$ is large we have that for sufficiently large $n$,
\begin{align*}
\frac{p_G(X)}{p_S(X)} &\leq 10\exp(n^{\alpha/2}).
\end{align*}
Therefore we can conclude
\[
\Prob_G(E\cap \Omega_{\rm good})
\leq 10\exp(n^{\alpha/2}) \int_E p_S(X)\diff X
\leq 10\exp(n^{\alpha/2}) \delta,
\]
as desired.
\end{proof}

We now go through the proofs of the lemmas, in reverse order.

\subsection*{Proof of Lemma~\ref{lem:gauss-conc}}

The key ingredient in the proof of Lemma~\ref{lem:gauss-conc} is the following classical concentration inequality.
\begin{lemma}
\label{lem:lipschitz}
Let $f:\Real^d\to\Real$ be a Lipschitz-continuous function, that is one satisfying
\[
|f(x) - f(y)| \leq L \|x-y\|,
\]
where the norm used above is the Euclidean one.
Let $X$ be a vector of independent
standard Gaussians, and set $\bar{f} = \Expec f(X)$.  Then
\[
\Prob(|f(X) - \Expec f(X)| \geq t) \leq \exp(-t^2/(2L^2))
\]
\end{lemma}

Now we can prove the proposition.
\begin{proof}[Proof of Lemma~\ref{lem:gauss-conc} using Lemma~\ref{lem:lipschitz}]
First we observe that $\lambda_{max}(X^TX)$ satisfies
\[
\sqrt{\lambda_{max}(X^TX)} = \sup_{\|u\|=\|v\|=1} u^TXv.
\]
We can think of $\sqrt{\lambda_{max}(X^TX)}$ as a function of $n^2$ independent Gaussian inputs,
and the Lipschitz constant is equal to the maximum Lipschitz constant of the functions
$u^TXv$.  This latter Lipschitz constant is given by
\[
\Big(\sum_{i,j=1}^n n^{-1} |u_i v_j|^2 \Big)^{1/2} = n^{-1/2} \|u\|\|v\| = n^{-1/2}.
\]
Thus $L\leq n^{-1/2}$.  Moreover for large enough $n$, $\Expec \lambda_{max}(X^TX) \leq 3$
(in fact, $\lim_{n\to\infty} \Expec \lambda_{max}(X^TX) = 2$).
Therefore
\[
\Prob(\sqrt{\lambda_{max}(X^TX)} \geq 3+t) \leq \exp(-nt^2/2).
\]
This concludes the proof of~\eqref{eq:max-conc}.

Now we prove~\eqref{eq:trac-conc}.  Let $g(t)$ be the function
\[
g(t)
= \begin{cases}
t^2 - 2t, & |t|\leq 4 \\
8, & t>4 \\
24, & t < -4.
\end{cases}
\]
Then for $X$ such that $\lambda_{max} (X^TX) \leq 4$ (which by the above occurs with probability
at least $1-\exp(-n/2)$,
\[
\phi(X) = \trace( (X^TX)^2 - 2(X^TX) ) = \trace(g(X^TX)) = \sum_i g(\lambda_i(X^TX)).
\]
Now let $E_{ij}$ be the matrix with a $1$ in the $(i,j)$ coordinate and $0$'s elsewhere.
Then the Lipschitz constant of $\phi(X)$ (as a function of the independent Gaussian matrix entries)
is
\[
L = n^{-1/2}\Big( \sum_{ij} \Big(\frac{d}{dt} \phi(X+tE_{ij})|_{t=0}\Big)^2\Big)^{1/2}
\]
But the derivative in this trace is given by
\[
\frac{d}{dt} \phi(X+tE_{ij})|_{t=0} = \trace(g'(X^TX) (X^TE_{ij} + E_{ji} X))
= (g'(X^TX)X^T)_{ij} + (Xg'(X^TX))_{ji}.
\]
Thus
\[
L \leq 2 n^{-1/2} (\trace[X^TX g'(X^TX)^2])^{1/2}.
\]
Since $g'$ itself is bounded by $24$, this means $\|g'(X^TX)\|\leq 24$, so
\[
L \leq 48 n^{-1/2} (\trace[X^TX])^{1/2}.
\]
If $\lambda_{max}(X^TX) \leq 4$, then $\trace[X^TX]\leq 4n$, so this becomes
\[
L \leq 96 \leq 100.
\]
And now~\eqref{eq:trac-conc} follows.
\end{proof}

\subsection*{Proof of Lemma~\ref{lem:pGS-compare}}
The proof of Lemma~\ref{lem:pGS-compare} is a relatively simple calculation.

For $X$ satisfying $\lambda_{max}(X) \leq m/n$,
\begin{align*}
\log(\frac{p_G(X)}{p_S(X)})
= \log(Z_S(n,m)/Z_G(n)) + \frac12 \sum_{i\in[n]} [(2n-m) \log(1-n\lambda_i/m) - n\lambda_i]
\end{align*}
To simplify this further we use the Taylor approximation
\[
|\log(1+t) - (t - \frac12 t^2)| \leq |t|^3,
\]
valid for $|t| \leq \frac{1}{10}$, which holds for $t=n\lambda_i/m$ when $\lambda_i\leq 2$
and $m>20n$.
Letting $E :=
|\log(\frac{p_G(X)}{p_S(X)}) - \log(c_G/c_S)|$, we can rearrange and cancel terms to obtain
\begin{align*}
|E| &=
\Big|\sum_{i\in[n]} [(2n-m) (- n\lambda_i/m - n^2\lambda_i^2/m^2 + O(n^3\lambda_i^3/m^3)) - n\lambda_i]\Big| \\
&\leq
\Big| \sum_{i\in [n]}n^2 \lambda_i^2 / m - 2n^2 \lambda_i/m\Big|
+ \sum_{i\in[n]} |2n^2 \lambda_i^2 /m^2|
+ \sum_{i\in [n]} |n^3 \lambda_i^3 /m^3| \\
&\leq
\frac{n^2}{m}|\trace[ (X^TX)^2 - 2X^TX] |
+ 8 \frac{n^3}{m^2} + 8 \frac{n^4}{m^3} \\
&\leq \sqrt{n} + O(1).
\end{align*}
To get to the last line we used that $\sum_{i\in[n]}\lambda_i^k = \trace[(X^TX)^k]$
and the following inequalities which hold for ``good'' matrices:
\begin{align*}
\lambda_i &\leq 2 \\
|\trace[ (X^TX)^2 - 2X^TX]| &\leq \sqrt{n}.
\end{align*}

\subsection*{Proof of Lemma~\ref{lem:Z-compare}}
First we establish some facts about submatrices of Haar-random orthogonal matrices.
First we need a calculation for the moments of such matrices.
\begin{lemma}[Lemma 2.5 of Jiang-Ma]
\label{lem:moments}
Letting $X$ be an $n\times n$ submatrix of an $m\times m$ orthogonal matrix,
scaled so that the entries have variance $n^{-1}$, we have
\begin{align*}
\Expec \trace[X^TX] &= n \\
\Expec \trace[(X^TX)^2] &= \frac{m}{m+2}[2n+1 - \frac{(n-1)^2}{(m-1)}].
\end{align*}
\end{lemma}

We combine this with the following concentration inequality on the orthogonal group.
The inequality below follows from the fact that the orthogonal group $SO(m)$ has Ricci curvature
$\frac{m-2}{4}$ which by the Bakry-Emery argument
(see~\cite{bakry2014analysis}, Theorem 2.1) shows that it has a log-Sobolev inequality with 
constant $\frac{4}{m-2}$,
and therefore Gaussian concentration for Lipschitz functions.
\begin{lemma}
\label{lem:orth-lip}
Let $f:\Real^{m^2} \to \Real$ be a function taking as input $m\times m$ matrices, and suppose that
$f$ has Lipschitz constant $L$, that is
\[
|f(X) - f(Y)| \leq L\|X-Y\|,
\]
where the norm used is the Hilbert-Schmidt norm, $\|X\|^2 = \sum_{ij} |x_{ij}|^2= \trace[X^TX]$.
Let $\bar{f} = \Expec f(X)$ where the expectation is over $X\in SO(m)$ sampled uniformly from the
Haar measure.  Then also over this probability measure we have
\[
\Prob( |f-\bar{f}| \geq t) \leq \exp(-mt^2/(8L^2)).
\]
\end{lemma}

Next we need to know some facts about submatrices of typical Haar-random matrices.
\begin{lemma}
\label{lem:typical-haar}
Let $X=\sqrt{m/n} Y$ where $Y$ is an $n\times n$ submatrix of a Haar-random orthogonal $m\times m$
matrix.  There exist $n_0$ such that for $n>n_0$ and $m>n^2$,
with probability at least $\frac12$, $X$ satisfies both $\lambda_{max}(X^TX) \leq 5$
and $\trace[(X^TX)^2 - 2X^TX] \leq 100$.
\end{lemma}
\begin{proof}
We use the fact that, for any sequence $m(n)$ satisfying  $m(n) \geq n^2$, we have
\[
\lim_{n\to\infty} \Expec \lambda_{max}(X^TX)^2 = 4.
\]
In particular for some $n_0>0$ we have for any $n>n_0$ and $m>n^2$ the inequality
\[
\Expec \lambda_{max}(X^TX)^2 \leq 5.
\]
Thus we conclude using Markov's inequality that
\[
\Prob(\lambda_{max}(X^TX) \geq 5)
\leq \frac{1}{25} \Expec(\lambda_{max}(X^TX)^2) \leq \frac15,
\]

Next, by Lemma~\ref{lem:moments} we have
\[
|\Expec \trace[(X^TX)^2 - 2X^TX]| \leq 1.
\]
Now let $g(t)$ be the same truncated version of $t^2-2t$ as in the proof of Lemma~\ref{lem:gauss-conc}.
The same argument as in there, combined with the concentration inequality of Lemma~\ref{lem:orth-lip}
implies that
\[
\Prob(|\trace[(X^TX)^2-2X^TX] | \geq 100) \leq \frac{1}{4}.
\]
Combining these with a union bound proves the result.
\end{proof}

We are finally ready to prove Lemma~\ref{lem:Z-compare}, and thus conclude the proof
of Proposition~\ref{prp:small-small}.
\begin{proof}[Proof of Lemma~\ref{lem:Z-compare}]
Let $A\subset\Real^{n\times n}$ be the set of matrices that are ``typical'' both for the Gaussian
distribution and as submatrices of orthogonal matrices:
\[
A := \{\lambda_{max}(X^TX) \leq 5\} \cap \{|\trace[(X^TX)^2 - 2X^TX]| \leq 100\}.
\]
Then by Lemma~\ref{lem:typical-haar} and also Lemma~\ref{lem:gauss-conc} we have
\[
\frac12 \leq \Prob_G(A) \leq 1
\]
and also
\[
\frac12 \leq \Prob_S(A) \leq 1.
\]
In particular,
\[
\frac12 \leq \frac{\Prob_G(A)}{\Prob_S(A)} \leq 2.
\]
Moreover, by Lemma~\ref{lem:pGS-compare} we have for $X\in A$ that
\[
|\log(p_G(X)/p_S(X)) - \log(c_G/c_S)| \leq 25 \frac{n^3}{m^2}
+ 125 \frac{n^4}{m^3} + 100 \frac{n^2}{m}  \leq 1
\]
for $m>101 n^2$ and $n$ sufficiently large.
Therefore
\begin{align*}
\frac{\Prob_G(A)}{\Prob_S(A)}
&= \frac{\int_A p_G(X)\diff X}{\int_A p_S(X)\diff X} \\
&= \frac{\int_A \frac{p_G(X)}{p_S(X)} p_S(X)\diff X}{\int_A p_S(X)\diff X} \\
&= \frac{\int_A \frac{p_G(X)}{p_S(X)} p_S(X)\diff X}{\int_A p_S(X)\diff X} \\
&= \frac{c_G}{c_S} \sup_{X\in A} \frac{c_Sp_G(X)}{c_Gp_S(X)}
\leq e\frac{c_G}{c_S} .
\end{align*}
Thus $\frac{c_G}{c_S} \geq \frac{1}{2e}$.  The argument also works to show
that $\frac{c_S}{c_G} \geq \frac{1}{2e}$, so the proof follows from $2e<10$.
\end{proof}

%% file: complex-squares.tex
In this section we quickly discuss the difficulties involved in proving variants of the square trick (Lemma~\ref{lem:square-trick-coex} in the context of coefficient extraction, and Lemma~\ref{lem:square-trick-extrapolation} in the context of extrapolation) in the case that $p$ and $q$ are \textit{complex} valued polynomials.
This is relevant if one wants to obtain hardness results for BosonSampling with a unitary random matrices (as opposed to orthogonal).  

The complex
case is significantly different from the real case because now one must recover a
complex phase from the unit circle in $\Complex$ rather than simply a sign $\pm1$ (for which there are only two possibilities -- this is used in Lemma~\ref{lem:square-trick-coex} for example).

This can be seen in the following example.
Let $q(t)=1$ be the constant polynomial, and let
\[
p_d(t) = \sum_{j=0}^d \frac{(it/2)^j}{j!}
\]
be the Taylor truncation of order $d$ of the exponential $e^{it/2}$.  Then by the Taylor
remainder formula,
\[
|p_d(t) - e^{it/2}| \leq \frac{2^{-d}}{d!}
\]
on the interval $[-1,1]$.  Therefore $||p_d|-1|\leq 2^{-d} (d!)^{-1}$ but also
for any phase $e^{i\theta}$ there exists $t\in[0,1]$ such that
$|p_d - e^{i\theta}| > 1/4$.

This example shows that there is no complex analogue of
Lemmas~\ref{lem:square-trick-coex} or \ref{lem:square-trick-extrapolation} if we only compare the values of $p$ and $q$ on
some real interval.  We can however perform an extrapolation by considering the values of the polynomial on the unit disk on the complex plane instead of the real interval $[-1,1]$.  To see that this has a hope of succeeding one can
see that
\[
\sup_{z\in B_1} ||p_d(z)| - 1| \gtrsim 1.
\]

We are unfortunately unable to provide a complete proof of a variant of Lemma~\ref{lem:square-trick-coex} in the complex case.  What we are missing is 
an interesting and seemingly difficult question in complex analysis \cite{mathoverflow1}.
\begin{conjecture}
\label{conj:similar-modulus}
For any complex-analytic functions $f$ and $g$ on the unit disk, there exists $\omega\in\Complex$, $|\omega|=1$ such that 
$|\omega|=1$ such that
\[
\max_{|z|\leq \frac14} |f(z) - \omega g(z)| \leq C\max_{|z|\leq 1} ||f(z)|-|g(z)||.
\]
Above $D_r$ is the complex disk of radius $r$ centered at the origin.
\end{conjecture}
Note that the restriction to $D_{1/4}$ (or at least some 
$D_r$ with $r<1$) is necessary for the conjecture to hold.
For example if $f=z^n$ and $g=z^{n+1}$ then for any $\omega=e^{i\phi}$, $\phi\in[0,2\pi]$
one has 
\[
\max_{|z|\leq 1} |z^n - \omega z^{n+1}|
\geq \max_{\theta\in[0,2\pi]} |e^{in\theta} - e^{i\phi} e^{i(n+1)\theta}| 
= \max_{\theta\in[0,2\pi]} |e^{i\theta} - e^{i\phi}|
= 2.
\]
On the other hand 
\[
\max_{|z|\leq 1}{||z|^n - |z|^{n+1}|}
\leq \max_{0\leq r\leq 1} r^n (1-r) 
\leq C n^{-1}.
\]

We are not claiming that Conjecture~\ref{conj:similar-modulus} is 
\textit{sufficient} to transfer our results to Haar-random unitary matrices, only that it seems to be necessary to overcome this obstacle before one can transfer our techniques to that setting.

%% file: main.bbl
\newcommand{\etalchar}[1]{$^{#1}$}